\documentclass[11pt,a4paper]{article}

\usepackage[utf8]{inputenc}
\usepackage{amssymb}
\usepackage{subfigure}
\usepackage{graphicx}
\usepackage{epsfig}
\usepackage{latexsym}
\usepackage{amsmath}
\usepackage[toc,page]{appendix}
\usepackage{comment}
\usepackage{color, xcolor}
\usepackage{setspace}

\usepackage{xspace}
\usepackage{algorithm}
\usepackage{algorithmic}
\usepackage{amsmath}
\usepackage{color,xcolor}
\usepackage{hyperref}
\usepackage{float}
\usepackage{enumerate}
\usepackage{dsfont}
\usepackage{amsmath,amssymb}
\usepackage{tabularx}
\usepackage{booktabs} 
\usepackage{pdflscape}
\usepackage{adjustbox}

\usepackage{graphicx}
\usepackage{algorithmic}
\usepackage{algorithm}
\usepackage{listings}
\usepackage[OT4,T1]{fontenc}
\usepackage{url}
\usepackage{multirow}
\DeclareMathOperator*{\argmin}{arg\min}

 \newcommand{\posA}{\vspace{0.08in}}

\newcommand\Eq{\mathcal{E}}
\newcommand\N{\mathbb{N}}

\newcommand{\tuple}[1]{\langle #1  \rangle}

\newcommand{\G}{{\cal G}}

\newcommand{\CidenticalJobs}{{\cal G}_{sym}}
\newcommand{\Ctwo}{{\cal G}_{two}}

\newcommand{\CidenticalMachines}{{\cal G}_{P}}
\newcommand{\Cglobal}{{\cal G}_{global}}

\newcommand{\CidenticalGlobal}{{\cal G}_{Pglobal}}
\newcommand{\CidenticalGlobalUniformRate}{{\cal G}_{Pglobal}^{+a}}
\newcommand{\CtwoP}{{\cal G}_{two,P}}

\newtheorem{remark}{Remark}[section]
\newtheorem{example}[remark]{Example}
\newtheorem{theorem}{Theorem}[section]

\newtheorem{lemma}[theorem]{Lemma}
\newtheorem{claim}[theorem]{Claim}

\newtheorem{definition}{Definition}[section]
\newtheorem{observation}[theorem]{Observation}

                      {}

\def\squarebox#1{\hbox to #1{\hfill\vbox to #1{\vfill}}}
\newcommand{\qed}{\hspace*{\fill}
\vbox{\hrule\hbox{\vrule\squarebox{.667em}\vrule}\hrule}\smallskip}

\def\eod{\vrule height 6pt width 5pt depth 0pt}
\newenvironment{proof}{\noindent {\bf Proof:} \hspace{.677em}}
                      {\hspace*{\fill}{\eod}}

\begin{document}

\title{
Coordination Mechanisms with Time-Dependent Processing Times
}
\author{Ido Borenstein\thanks{School of Computer Science, Reichman University, Israel, E-mail:\href{mailto:idobor@gmail.com}{\texttt{idobor@gmail.com}}, ORCID:\href{https://orcid.org/0009-0006-2268-5591} {0009-0006-2268-5591}} \and 
Tami Tamir\thanks{School of Computer Science, Reichman University, Israel. E-mail: \href{mailto:tami@runi.ac.il}{\texttt{tami@runi.ac.il}}, ORCID: \href{https://orcid.org/0000-0002-8409-562X}{0000-0002-8409-562X}}}
\date{}
\maketitle

\begin{abstract}
Job-scheduling games serve as a fundamental model in algorithmic game theory for understanding decentralized resource allocation among self-interested agents. To reduce equilibrium inefficiency in these settings, system designers employ coordination mechanisms - local scheduling policies that dictate the order of processing. Traditionally, the analysis of job scheduling games has assumed fixed processing times. However, in many realistic environments, ranging from cyber-security response to high-frequency trading, a job's length depends on its starting time. We study coordination mechanisms with time-dependent processing times, where job lengths are linear functions of their start times, exhibiting either positive deterioration (increasing length) or negative deterioration (decreasing length). We analyze these games under various local policies. By introducing the concept of delay-averse agents, we provide a unifying framework to characterize equilibrium existence. For delay-averse jobs, we show that stability is maintained and pure Nash equilibria (NE) can be computed efficiently. In contrast, for non-delay-averse jobs, we demonstrate that a NE may not exist, and prove that deciding its existence is NP-complete, even on $m$ identical machines -- a fundamental departure from classical coordination mechanisms.

Regarding equilibrium inefficiency, we show that the Price of Anarchy (PoA) can be significantly higher than in environments with fixed processing times. To mitigate this, we propose and analyze three coordination mechanisms: SBPT (Shortest Basic Processing Time), which reduces the PoA in games with positive deterioration to a constant, and SDR (Smallest Deterioration Rate) and LBDR (Largest Basic-Deterioration Ratio) for negative deterioration, which achieve tight constant PoA bounds of $2$ and $\max\{\frac{e}{e-1}, 2-\frac{1}{m}\}$, respectively. Our results bridge the gap between centralized time-dependent scheduling and decentralized game-theoretic analysis.
\end{abstract}

\section{Introduction}
Scheduling environments encountered in practice are often inherently decentralized: jobs are controlled by self-interested agents, and no central authority exists to globally optimize their assignment to machines. Such settings are naturally analyzed using game-theoretic tools where stable outcomes are captured by equilibria.
Job scheduling games arise in applications such as production lines, traffic and communication networks. In these settings, each job acts as a strategic player that selects a machine to minimize its completion time.  
System designers often employ coordination mechanisms~\cite{CKN09}, which enforce local scheduling policies on each machine (e.g., Shortest First).

In standard models, jobs have fixed processing times. However, in many environments, the processing time depends on the starting time. Such \emph{time-dependent} or \emph{deteriorating} jobs have been studied intensively in centralized settings, but received limited attention from a game-theoretic perspective. In this work, we bridge this gap and study coordination mechanisms for jobs with time-dependent processing times. Our model is general, allowing heterogeneous machine policies and both increasing and non-increasing processing times.

Our work is motivated by applications in which delays affect task complexity or resource efficiency, such as cyber-security response, data cleaning, and medical treatment. In some cases, delays increase processing times, while in others, processing times decrease due to {\em learning effects} or improved efficiency. While the scheduling community has extensively studied time-dependent processing times in centralized models (see Section~\ref{sec:related}), there has been very little research into how these dynamics behave in strategic environments. 

We assume that jobs are controlled by self-interested agents and analyze the corresponding games under various coordination mechanisms and priority policies. By introducing the concept of delay-averse jobs, we provide a complete picture of equilibrium existence, computation, convergence of natural dynamics, and equilibrium inefficiency. 
We distinguish between classes of instances in order to provide tight bounds for specific classes. As we show, in general, the Price of Anarchy (PoA) can be significantly higher than in environments with fixed processing times. To mitigate this, we propose and analyze several coordination mechanisms with a guaranteed constant PoA.

\subsection{Notations and Problem Statement}
An instance of a {\em scheduling game with machine-dependent priority lists and time-dependent processing times} is given by a tuple $G=\tuple{N,M,(p_i)_{i \in N},(s_j)_{j \in M},(\pi_j)_{j \in M}}$, where $N$ is a finite set of $n\geq 1$ jobs, each corresponding to a single player, $M$ is a finite set of $m\geq 1$ machines, $p_i$ is the processing time {\em function} of job $i$. That is, $p_i(t)$ is the length of job $i$ if it starts being processed at time $t$. 
$s_j\in\mathbb{R}^+$ denotes the speed of machine $j\in M$, and $\pi_j:N \rightarrow \{1,\ldots,n\}$ is a bijection describing the {\em priority list} of machine $j\in M$. 
For a machine $j\in M$ and two jobs $u,v\in N$, we say that $u$ has higher priority than $v$ on machine $j$, iff $\pi_j(u)<\pi_j(v)$.

We assume that the processing-time function is linear in $t$. Specifically, every job $i$ is associated with three parameters: $b_i \ge 0$, its \emph{basic processing time}; $a_i \ge 0$, its \emph{deterioration rate}; and $v_i \in \{-1, 1\}$, its {\em deterioration sign}. The processing-time function is defined as $p_i(t) = b_i + v_i \cdot a_i \cdot t$. 
A positive deterioration sign ($v_i=1$) corresponds to settings where job length increases over time. A negative deterioration sign ($v_i=-1$) corresponds to settings where job length decreases, in which case the job is also assigned a \emph{threshold processing time} $\tau_i>0$, representing its minimal possible length. Formally, for $v_i=-1$, $p_i(t) = \max\{\tau_i, b_i - a_i \cdot t\}$. 
This threshold captures the fact that, even if a job’s effective size decreases over time, there is always a minimal amount of work that cannot be reduced further.
For simplicity, we omit $v_i$ when the context is clear. In particular, the processing time
function of a job with a positive deterioration rate is simply $p_i(t) = b_i + a_i \cdot t$.

A profile of the game is denoted {\em a schedule} and is defined by the players' strategy profile $\sigma=(\sigma_i)_{i\in N}\in M^N$. The strategy of job $i \in N$, denoted $\sigma_i$, is the machine selected by job $i$. Given a strategy profile $\sigma$, the jobs are processed according to their order in the machines' priority lists. 
The starting and completion times of job $i$ in a schedule $\sigma$ are denoted $S_i(\sigma)$ and $C_i(\sigma)$, respectively, and are defined as follows. Denote by $i_1,i_2,\ldots$ the jobs processed on machine $j$, according to their order in $\pi_j$, that is, $\pi_{j}(i_1) \leq \pi_{j}(i_2) \leq \ldots$.
For job $i_1$, $S_{i_1}(\sigma)=0$ and $C_{i_1}(\sigma)=p_{i_1}(0)/s_j$. For job $i_k$, $S_{i_k}(\sigma)=C_{i_{k-1}}(\sigma)$, and $C_{i_k}(\sigma)=S_{i_k}(\sigma)+p_{i_k}(S_{i_k}(\sigma))/s_j$. We use $p_i^{\sigma}$ to denote the processing time of job $i$ in a profile $\sigma$, that is, $p_i^{\sigma}=p_i(S_i(\sigma))$. In our game, the cost of job $i$ is defined to be its completion time, that is, $cost_i(\sigma)=C_i(\sigma)$. Thus, every job chooses a machine aiming at minimizing its completion time.

For a profile $\sigma$, denote by $(\sigma_i^{\prime},\sigma_{-i})$ the profile in which player $i$ modifies its strategy to 
 $\sigma_i^{\prime}$ and all other players keep their strategy as in $\sigma$. A strategy profile $\sigma$ is a {\em pure Nash equilibrium (NE)} if for all $i\in N$ and all $\sigma_i^{\prime}\in M$, 
we have that $cost_i(\sigma)\leq cost_i(\sigma_i^{\prime},\sigma_{-i})$. Let $\Eq(G)$ denote the set of pure Nash equilibria for a given instance $G$. We note that $\Eq(G)$ might be empty.

For a profile $\sigma$, let $cost(\sigma)$ denote the social cost of $\sigma$. The social cost is defined with respect to some objective, in our case, the makespan, i.e., $C_{max}(\sigma):=\max_{i\in N} C_i(\sigma)$.
It is well known that decentralized decision-making may lead to sub-optimal solutions from the
point of view of the society as a whole. For a game $G$, let $P(G)$ be the set of feasible profiles of $G$. We denote by $OPT(G)$ the cost of a social optimal solution, i.e., $OPT(G)=\min_{\sigma \in P(G)} cost(\sigma)$. We quantify the
inefficiency incurred due to self-interested behavior according to the
\emph{price of anarchy} ($PoA$) \cite{Koutsoupias:1999:WE:1764891.1764944}, and \emph{price of
stability} ($PoS$) \cite{AD+08}. The $PoA$ is the worst-case inefficiency of a pure Nash equilibrium, while the $PoS$ measures the best-case inefficiency of a pure Nash equilibrium. 

\begin{definition}
\label{def:ineff}
Let $\G$ be a class of games, and let $G$ be a game in $\mathcal{G}$.
Let $\Eq(G)$ be the set of pure Nash equilibria of the game $G$. Assume that $\Eq(G) \neq \emptyset$.
The {\em price of anarchy} of $G$ is the ratio between the
\emph{maximum} cost of a $NE$ and the social optimum of
$G$, i.e.,
$\mbox{PoA}(G) = \max\limits_{\sigma\in \Eq(G)} cost(\sigma)/\mbox{OPT}(G)$.
The {\em price of anarchy} of $\mathcal{G}$
is $\mbox{PoA}(\mathcal{G}) = sup_{ G\in \mathcal{G}}\mbox{PoA}(G)$.
\end{definition}

\emph{Best-Response Dynamics} (BRD) is a natural method by which players proceed toward a $NE$ via a sequence of improving deviations. The question of BRD convergence and the quality of possible BRD outcomes are major questions in the study of job scheduling games~\cite{Ros73,FST17}.

\noindent{\bf Optimal Starting Time of a Job:}
Since the cost of a job is defined by its completion time, it seems natural that jobs would prefer starting their processing as early as possible. This underlying intuition becomes incomplete in an environment with negative deterioration rate, where a later start may lead to a shorter duration, and consequently, to an earlier completion time.
Consider an assignment of job $i$ with rate $a_i$. Denote by $c_i(t,s)$ the completion time of job $i$ as a function of its starting time $t$ and the speed $s$ of the machine it is assigned on. 
The following straightforward observations provide the optimal starting time of a job. 

\begin{observation}
    \label{ob:optimalStartTime1}
For $p_i(t) = b_i + a_i\cdot t$, where $b_i \ge 0$, and $a_i \ge 0$, the completion time function $c_i(t,s)$ gets its minimal possible value at time $t=0$. 
\end{observation}

\begin{observation}
    \label{ob:optimalStartTime}
For $p_i(t) = \max(\tau_i,b_i-a_i\cdot t)$, where $b_i \ge 0, b_i>\tau_i \ge 0$, and $a_i > 0$, the completion time function $c_i(t,s)$ gets its minimal possible value $(i)$ at time $t=0$, if $a_i<s$, $(ii)$ at time $t=\frac{b_i-\tau_i}{a_i}$ if $a_i>s$, and $(iii)$ at any time $0 \le t\le\frac{b_i-\tau_i}{a_i}$ if $a_i=s$.
\end{observation}

We assume that in case of a tie, where a job has the same completion time on several machines, it prefers the one on which its processing time is longer. Such a scenario is relevant when a job has a negative deterioration that equals the machine's speed, or if machines have different speeds. Our results are valid independent of this tie-breaking rule.

We say that a job $i$ is {\em delay-averse} if its completion time increases with its starting time, independent of the machine it is assigned on. Formally, $i$ is delay-averse if $v_i=1$ (that is, $i$ has a positive deterioration rate), or if $v_i=-1$ and $a_i \le s_{min}$, implying that $c_i(t,s)$ is a non-decreasing function for any $s_j, j \in M$.
As we will show, the presence of jobs that are not delay-averse is the core reason for the non-stability of our game. 

\posA\noindent{\bf Restricted Game Classes:}
Some of our results refer to restricted classes of games, characterized by properties of both the jobs and the machines.
Specifically, we consider the following classes of games:
\begin{description}
\item $\CidenticalJobs$: Games with symmetric (identical) jobs, all  having the same processing-time function.
\item $\Ctwo$: Games played on two related machines.
\item $\CidenticalMachines$: Games played on identical (same speed) parallel machines.
\item $\Cglobal$: Games with a {\em global priority list}: all machines have the same priority list, $\pi$.
\item $\CidenticalGlobalUniformRate$: Games with positive deterioration on identical machines with a global priority list, where all jobs have the same deterioration rate $a$.
\end{description}

For each $h$ , let $\G{^+_h}$, $\G{^-_h}$ denote the corresponding class assuming jobs with positive or negative deterioration rate, respectively. 
Similarly, for each $h$ , let $\G{^{DA}_h}$ denote the corresponding class assuming all jobs are {\em delay-averse}.
In addition, let $\G_h^{-DA}= \G{^{DA}_h} \cap \G{^-_h}$, that is, games in which all jobs are delay-averse with a negative deterioration rate.

\subsection{Related Work}
\label{sec:related}
Our work lies at the intersection of scheduling with time-dependent processing times and job-scheduling games under coordination mechanisms. We briefly review each line of work and the limited studies that combine them.

Early work on time-dependent processing times focused on centralized settings. The paper~\cite{GuptaGuptaSingel} studied positive deteriorating jobs on a single machine, showing that in general, where processing time functions are non-linear, minimizing the makespan is NP-complete, while with linear functions, ordering the jobs by $b_i/a_i$ is optimal. Subsequent work considered various extensions, including stochastic models and general processing-time functions~\cite{BrowneYechieaRandom,AlidaeeWormerSurvey}, as well as simpler structured models that admit optimal policies for the sum of completion times objective~\cite{MosheiovFixedBasicVshapred,MoeshiovRankBased}. 

For decreasing processing times, several variants were shown to be NP-complete~\cite{HoDecreasingWithDeadlines,BJ00}, including settings that focus on feasibility under release times and deadlines. Special cases admit efficient solutions, for example when $a_i$ and $b_i$ are related~\cite{WangXiaSpecialCase}, or for specific objectives such as total completion time~\cite{WangSquared}. The case of scheduling deteriorating jobs on identical parallel machines was studied in~\cite{KangNgParallel}. They developed a fully polynomial approximation scheme (FPTAS) for makespan minimization, in which each machine processes the jobs in non-decreasing order of $b_i/a_i$. Unlike our model, in which jobs are associated with threshold processing times, in previous models every job is associated with a deadline, and reaching a negative processing time is mathematically possible only after the deadline.

Scheduling games with fixed processing times have been widely studied, focusing on the inefficiency of equilibria~\cite{Koutsoupias:1999:WE:1764891.1764944}, with tight bounds known for various machine models~\cite{CzumajV07,AART06,GLM10}. To mitigate inefficiency, coordination mechanisms were introduced in~\cite{CKN09} and further studied in~\cite{ILMS09,AJM15,RST21}, which analyzes standard priority policies such as SPT and LPT, as well as arbitrary priority lists.

Only few works consider strategic behavior with time-dependent processing. A simple linear model for positive deterioration is studied in~\cite{LiLiuLiFirstGameTheory}, showing convergence of natural dynamics to a pure Nash Equilibrium, and analyzing the $PoA$. Coordination mechanisms such as LDR and SDR were analyzed in~\cite{CLTY17}, proving existence of equilibria and bounds on inefficiency. The model was extended to parallel batch machines in~\cite{YuBatchProcessing}.

\subsection{Our Contributions}
We summarize our main results on equilibrium existence, convergence, computational complexity, and inefficiency in scheduling games with time-dependent processing times.

Section~\ref{sec:eqExistence} examines equilibrium existence and computation. We prove that greedy algorithms compute a NE for games in $\CidenticalJobs$ and $\CidenticalMachines^{DA}$. Additionally, we provide a specific algorithm for $\Ctwo^{DA}$ and prove that an adapted List-Scheduling algorithm guarantees a NE for games in $\Cglobal$.
We then study convergence of Best-Response Dynamics (BRD) and prove that every best-response sequence converges to a $NE$ for every instance $G \in \CidenticalJobs \cup \Ctwo^+ \cup \CidenticalMachines^{DA} \cup \Cglobal$.

In Section~\ref{sec:NDA}, we study games with non-delay-averse jobs, which may benefit from delaying their start times. We demonstrate that a NE is not guaranteed to exist in these settings and provide a tight characterization of instances that do admit a NE, considering both the number of players/machines and the environment (identical vs. related machines).
We further analyze the computational complexity of deciding NE existence, proving the problem is NP-complete, even on identical machines. This result highlights a fundamental departure from  games with fixed processing times, where a pure NE is  guaranteed to exist on identical machines~\cite{RST21}.

In Section~\ref{sec:EqInefficiency} we analyze equilibrium inefficiency with respect to the makespan objective for classes in which $NE$ existence is guaranteed. We begin by proving that the Price of Anarchy ($PoA$) of games in $\CidenticalJobs$ is $1$.
We then study games in $\CidenticalGlobalUniformRate$ (positive deterioration, global priority list, uniform rate $a$). 
We provide a tight exponential bound for an arbitrary number of machines, showing that the $PoA$ can grow as $(1+a)^{\frac{n}{m}}$.

We next turn to games with negative deterioration, restricting attention to instances with delay-averse jobs (for which $NE$ existence is guaranteed). For identical machines with arbitrary priority lists, we prove a tight bound $PoA(\CidenticalMachines^{-DA})=3-\frac{1}{m}$.  For identical machines with a global priority list we prove a bit lower tight bound, $PoA(\CidenticalGlobal^{-DA})=3-\frac{2}{m}$. Our methodology relies on establishing a bound on
the ratio of total processing times in optimal versus arbitrary schedules.

In Section~\ref{sec:efficient} we introduce coordination mechanisms aimed at reducing equilibrium inefficiency. 
The primary challenge when analyzing jobs with non-fixed durations is that traditional bounds based on job lengths are no longer valid. For example, two lower bounds on the makespan that are used heavily in job-scheduling literature are $\sum_i p_i/m$ and $p_{max}$, these terms are not even defined in our model.  
To address these challenges, we introduce three policies. Two for jobs with negative deterioration rates, and one for jobs with a uniform positive deterioration rate. 

$(i)$ SDR (Smallest Deterioration Rate) policy, prioritizes jobs in non-decreasing order of $a_i$. In its analysis we use {\em work-based density analysis}: we shift the analytical focus from the standard time axis to a cumulative work axis. We define the work-based deterioration density, $\mathcal{A}_{\sigma}(w)$, representing the aggregate deterioration rates of all waiting jobs at any workload level $w$. We prove that the SDR policy acts as a density maximizer, ensuring that the integral of waiting deterioration is maximized relative to machine effort. By coupling this with a bound on the total processing times, we establish a tight constant bound of $2$ on the Price of Anarchy.

$(ii)$ LBDR (Largest Basic-deterioration Ratio) policy, prioritizes jobs based on the ratio $R_i = b_i/a_i$). In its analysis, we employ a product-based inductive technique. We relate the completion time of any machine in an optimal schedule to the parameters of the last job in the LBDR sequence, showing that the term $1 - C_i/R_n$ is bounded by the product of terms $(1 - b_h/R_n)$ for all jobs $h$ on that machine. We then express the starting time of the last job in the $NE$ profile as a convex function of $b_n$, which achieves its maximal values at $b_n=0$ and $b_n=OPT$. This framework successfully yields a tight constant $PoA$ of $e/(e-1)$ for $m \le 3$ and $2-1/m$ for larger systems.

$(iii)$ For jobs with a uniform positive deterioration rate $a$, SBPT policy (Shortest Basic Processing Time) addresses the main cause of exponential inefficiency: cascading delays across a sequence of jobs. We employ a round-based optimality argument, proving that in any Nash Equilibrium under SBPT, jobs are effectively partitioned into "rounds" (sets of $m$ jobs), where each round contains the shortest available jobs assigned across the $m$ machines. We prove that this configuration minimizes the total processing time of the entire system, reducing the potential PoA from $(1+a)^{n/m}$ to a constant bound of $\frac{2m+am-1}{m+a}$.
Due to space constraints, some of the proofs are given in the appendix.

\section{Equilibrium Existence and Computation}
\label{sec:eqExistence}
For the four classes $\CidenticalJobs$,
$\Ctwo$, $\CidenticalMachines$, $\Cglobal$, it is known that a $NE$ is guaranteed to exist if job lengths are fixed~\cite{RST21}. 
Clearly, this is a special case of our setting (where for all $i$, $a_i=0$). 
Our goal in this section is to examine whether the existence of a $NE$ is hurt when jobs lengths vary with time. As we show, the answer is not uniform and depends on the deterioration rates of the jobs.
Our first results are based on adaptations of the algorithms in~\cite{RST21} for fixed-length jobs. 

\begin{theorem}
\label{thm:unitjobs}
Every game $G \in \CidenticalJobs$ has a NE, which can be computed in poly-time.\hfill\qed
\end{theorem}
\begin{theorem}
\label{thm:m2NE}
Every game $G \in\Ctwo^{DA}$ has a NE, which can be computed in poly-time.\hfill\qed
\end{theorem}

We turn to consider the class $\Cglobal$ of games with a global priority list. 
The following is an adaptation of Graham's List-Scheduling (LS) algorithm~\cite{Gra66}. The jobs are considered according to their order in $\pi$, each assigned on a machine that minimizes its completion time.
Assume w.l.o.g., that $\pi=(1,2,\ldots,n)$. 
\begin{algorithm}[H]
\caption{ List-Scheduling Algorithm for games with a global priority list} \label{alg:LS}
\begin{algorithmic}[1]
\STATE Let $\ell_j$ denote the completion time of machine $j$. Initially, $\ell_j=0$ for all $1 \le j \le m$
\FOR {$i= 1$ to $n$}
\STATE Let $j^{\star} = \argmin_j ~\ell_j+p_i(\ell_j)/s_j$
\STATE Assign job $i$ on machine $j^{\star}$
\STATE $\ell_{j^{\star}} = \ell_{j^{\star}}+p_i(\ell_{j^{\star}})/s_{j^{\star}}$.
\ENDFOR
\end{algorithmic}
\end{algorithm}

\begin{theorem}
    \label{thm:ls_on_global}
   Let $G \in \Cglobal$. The set of NE schedules of $G$ equals the set of possible outcomes of List-scheduling  algorithm.\hfill\qed
\end{theorem}

The last class for which we prove NE existence includes games played on identical machines. The following natural greedy algorithm is for instances with arbitrary (non-global) priority lists: Starting with an empty schedule, perform $n$ iterations, where in each iteration one job is assigned. Specifically, for every machine, $j$, the completion time of the next unassigned job in $\pi_j$ is calculated, and the job that achieves a minimal completion time among these candidates is assigned.

\begin{theorem}
\label{thm:identicalMachinesNE}
Every game $G \in \CidenticalMachines^{DA}$ has a NE, which can be computed in polynomial time.\hfill\qed
\end{theorem}

To emphasize the relevance of List-Scheduling (LS) algorithm to our model, we note that even without a global priority list, if the jobs are considered according to their start time in a NE profile $\sigma$, we get the following:
\begin{observation}
    \label{obs:neForIdenticalMachineIsLs}
    Every $NE$ schedule of a game in $\CidenticalMachines^{DA}$ is a possible outcome of LS algorithm. \hfill\qed
\end{observation}

Finally, we show that every sequence of best-response dynamics (BRD) converges to a $NE$ for every class in which a $NE$ is guaranteed to exist.
\begin{theorem} 
\label{th:brdcycle}
Let $G$ be a game instance in $\CidenticalJobs \cup \Ctwo^+ \cup \CidenticalMachines^+ \cup \CidenticalMachines^{DA} \cup \Cglobal$. Any best-response sequence  
in $G$ converges to a $NE$. \hfill\qed
\end{theorem}

\section{Games with Non Delay-Averse Jobs} 
\label{sec:NDA}
In this section we consider instances in which all jobs have non-increasing processing time function, formally, for all $i$, $v_i=-1$.
In addition, we do not assume that the jobs are delay-averse. Thus, jobs' deterioration rates might be higher than the machines' speeds. We show that a Nash Equilibrium need not exist in this setting even for very small instances and that the problem of deciding whether a game has a NE profile is NP-hard, even if only three jobs have time-dependent processing times.

We first show that a $NE$ may not exist even in a simple game of two jobs and two machines.

\begin{example}
\label{ex:noNe2}
{\em
Consider the game $G_{noNE2} \in \Ctwo^-$ with $N=\{u,v\}$ and $M=\{M_1,M_2\}$, where $\pi_1=(u,v)$ and $\pi_2=(v,u)$. The speeds of the machines are $s_1=1$, and $s_2=0.5$. The jobs' basic processing times are $b_u=24$, $b_v=8$, and both have negative deterioration rate $a=1.5$ and threshold processing time $\tau=3$, yielding the processing-time functions $p_u = max(24-1.5\cdot t,3)$ and $p_v = max(8-1.5\cdot t,3)$. Since $m=2$, each job has two strategies. The jobs' costs in the four possible profiles are given in Table~\ref{tab:noNE2}.
As can be seen in the table, the jobs will benefit from deviating in clockwise direction, thus, $G_{noNE2}$ has no pure Nash equilibrium.
}
\begin{table}[htbp]
 \begin{center}
  \begin{tabular}{|c||c|c|}
  \hline
  {u $\char`\\$ v} & $M_1$ & $M_2$\\
  \hline\hline
  $M_1$   & $\mathbf{24}, 24+3=\mathbf{27}$ & $\mathbf{24, 16}$ \\
  \hline
  $M_2$  & $24/0.5=\mathbf{48, 8}$ & $16 + 3/0.5 = \mathbf{22, 16}$ \\
  \hline
  \end{tabular}
 \end{center}
 \caption{Jobs' completion times in the game $G_{noNE2}$. Every entry gives the completion time of job $u$ followed by the completion time of job $v$ in the corresponding profile.}
  \label{tab:noNE2}
\end{table}
\end{example}

The game $G_{noNE2}$ is played on two related machines, that have different speeds. For two identical machines we show that every two-player game has a $NE$ profile, while a $NE$ may not exists in a three-player game.

\begin{theorem}
\label{thm:dwin2M2}
Every two-player game $G \in \CtwoP^-$ has a $NE$ profile.   \hfill\qed
\end{theorem}
\begin{theorem}
\label{thm:dwin2M3}
There exists a three-player game $G_{noNE3} \in\CtwoP^-$ that has no $NE$ profile.  
\end{theorem}
\begin{proof}
For $n=3$, let $G_{noNE3}$ be a game with $N=\{u,v,w\}$ and $M=\{M_1,M_2\}$, both having speed $1$. The jobs have the following processing time functions: $$p_u(t)=max\{5-1.05\cdot t,0.2\}, p_v(t)=max\{4-1.1\cdot t,0.2\}, p_w(t)=max\{3-1.2\cdot t,0.2\}.$$ 
The priorities of the machines are $\pi_1=(v,u,w)$ and $\pi_2=(w,u,v)$. 
With two machines and three jobs, the game has $8$ possible game profiles. As demonstrated in Figure~\ref{fig:noNE3Jobs} and detailed in the appendix, none of these profiles is a $NE$.
\end{proof}

\begin{figure}[h]
\centering
\includegraphics[width=1\textwidth]{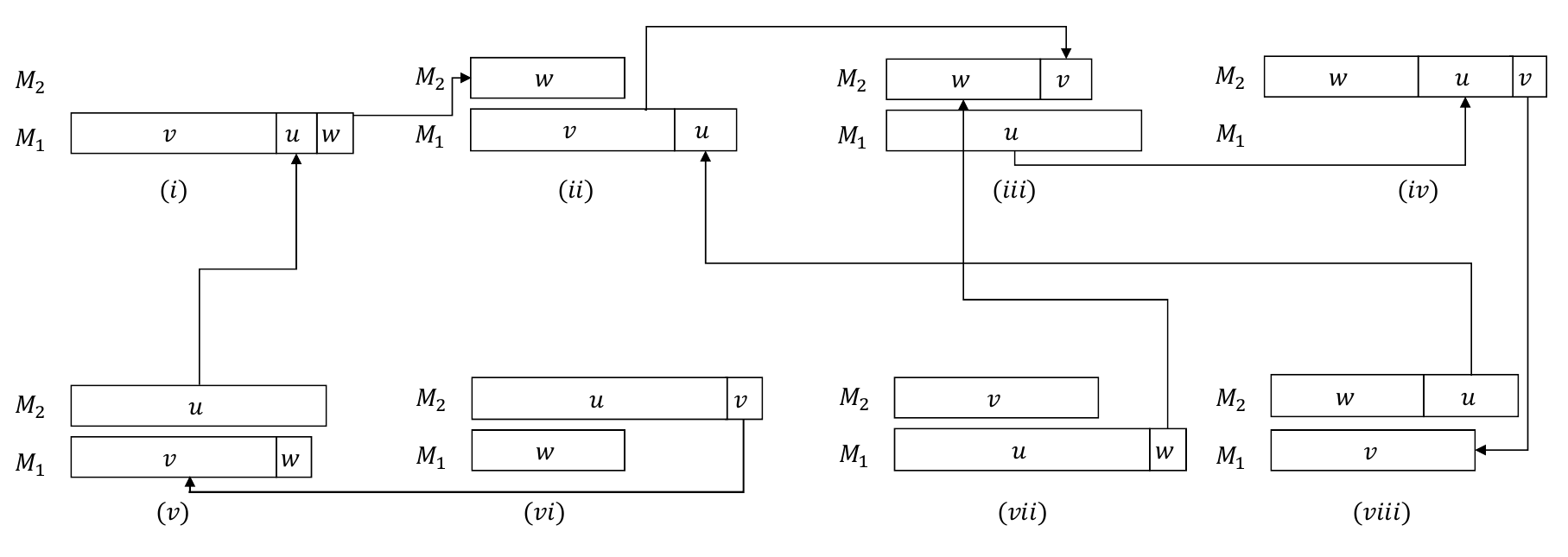}
\caption{The possible profiles of the $G_{noNE3}$. No profile is a $NE$.}
\label{fig:noNE3Jobs}
\end{figure}

We turn to consider the computational complexity of deciding whether a given game has NE profile. For games with fixed processing times, the problem was studied in~\cite{RST21}, where it was shown to be NP-complete for games with fixed processing times and related machines.  For identical machines and fixed processing times, a $NE$ is guaranteed to exist. We show that when jobs have negative deterioration, the problem is NP-complete even on identical machines. 
The hardness proof is by a reduction from $3$-bounded $3$-dimensional matching ($3$DM-$3$). We construct a gadget centered around the game $G_{noNE3}$. The reduction is designed such that this unstable sub-game is "triggered" if and only if a perfect matching does not exist in the $3$DM-$3$ instance.  

\begin{theorem}
    \label{th:neISNP}
    Given a game $G\in G_P^-$, it is NP-complete to decide whether $G$ has a $NE$. \hfill\qed
\end{theorem}

\section{Equilibrium Inefficiency}
\label{sec:EqInefficiency}

A common measure for evaluating the quality of a schedule is the makespan, 
given by $C_{max}(\sigma) = \max_{i \in N} C_i(\sigma)$.
In this section we analyze the equilibrium inefficiency with respect to the makespan, for each of the classes for which a $NE$ is guaranteed to exist.

\subsection{Games with Symmetric Jobs}
\label{sec:ineffSym}
We begin with $\CidenticalJobs$, the class of instances with symmetric jobs with positive or negative deterioration rates (in particular, non delay-averse). We show that every NE profile achieves minimum makespan, even with machines having different speeds and arbitrary priority lists. 
\begin{theorem}
\label{th:poaSym}
$PoA(\CidenticalJobs)= 1$.
\end{theorem}

\subsection{Games with Positive Deterioration Rates}

In standard coordination mechanisms with fixed processing times, that admit a $NE$, the $PoA$ is known to be a polynomial function of the number of jobs and machines~\cite{RST21}. On the contrary, if the jobs' processing time function is $p_i(t)=a_i\cdot t$ (that is, for all $i$, $b_i=0$) the PoA can grow exponentially with the number of jobs or machines (\cite{LiLiuLiFirstGameTheory}). We show that the exponential inefficiency persists when jobs have arbitrary basic processing times. 

We first show that the $PoA$ is not bounded by a function of $m$, for any $m$ identical machines having a global priority list. This is in contrast to fixed-length jobs, where  $PoA=\Theta(\log m)$~\cite{Aspnes}. 
\begin{theorem}
\label{thm:PoAnotr}
For any $m>1$, given $r>1.5$, there exists a game $G_{m,r}$ on $m$ identical machines and $m+1$ jobs, such that $PoA(G_{m,r})= r$.
\end{theorem}

\noindent{\bf Uniform Deterioration Rate, Global Priority List:}
Let $\CidenticalGlobalUniformRate$ be the class of games played on identical machines, with a global priority list, and jobs with a uniform positive deterioration rate. We show that, unfortunately,  the $PoA$ of a game in this class may grow exponentially with $n/m$. This extremely high $PoA$ may be achieved due to high {\em cascading delays}. As in fixed-length scheduling, the starting time of a job is influenced by the durations of all preceding jobs on the same machine. In models with positive deterioration, this effect is not just additive, but multiplicative, leading to exponential growth in the total workload. 
The proof of the following tight bound is given in the appendix. 
\begin{theorem}
\label{thm:UniformRateGlobalCombined}
$PoA({\CidenticalGlobalUniformRate})=\Theta((1+a)^{\frac{n}{m}})$.
\end{theorem}

The lower bound is attained by the following game. It has $mk$ jobs $x_i$ with $p_{x_i}(t)=1+at$, and $m+1$ long jobs $y_i$ with $p_{y_i}(t)=\frac{w-1}{a}+at$, where $w=(1+a)^{mk}$. The global list is $(y_0,\ldots,y_{m-1},x_0,\ldots,x_{mk-1},y_m)$. In an optimal schedule the $x$-jobs start at time $0$ and remain short. In the NE every $x$-job follows a long $y$-job, and the delays cascade (see Figure~\ref{fig:unboundedPoAFig}).

\begin{figure}[h]
\centering
\includegraphics[width=1\textwidth]{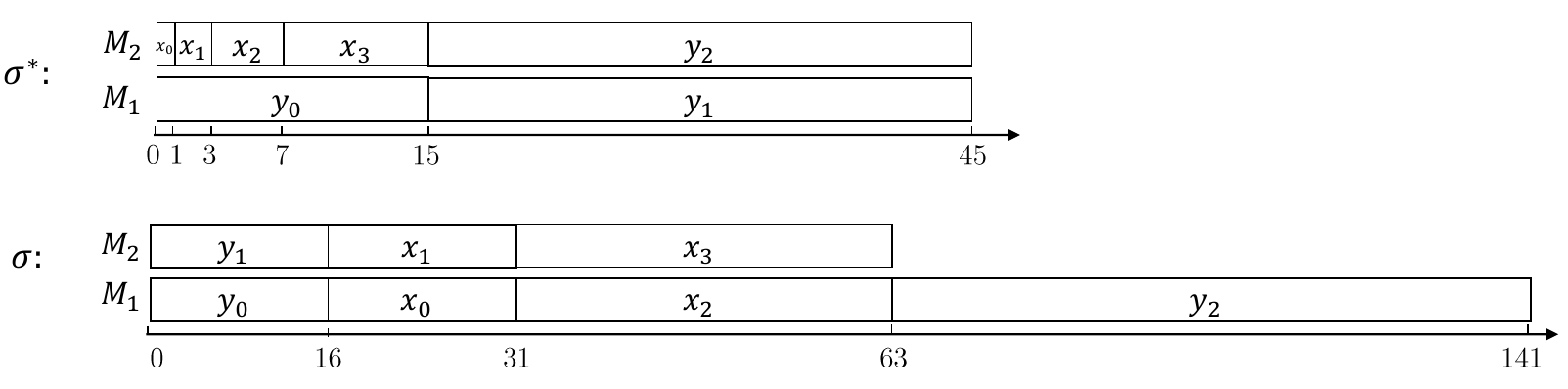}
\caption{A game corresponding to $n=7, m=2$ and $a=1$. Let $w=(1+a)^4=16$. The global priority list is $\pi = (y_0,y_1,x_0,x_1,x_2,x_3,y_2)$. $\sigma^*$ is an optimal schedule with makespan $3w-3=45$; $\sigma$ is a $NE$ schedule with makespan $2w\sqrt{w}+w-3=141$.}
\label{fig:unboundedPoAFig}
\end{figure}

\subsection{Games with Identical Machines and Negative Deterioration Rates}

Unlike games with positive deterioration rate, we show that the PoA of games with negative deterioration rates is bounded by a constant. Our analysis refers to games with delay-averse jobs, for which $NE$ existence is guaranteed.
Our methodology relies on establishing a bound on the ratio of total processing times in optimal schedules. We show that when the ratio between the sums is maximal, so is the $PoA$.

\subsubsection{Arbitrary Priority Lists} 

We first analyze games with identical machines and arbitrary priority lists. With fixed-time jobs, it is known that the $PoA$ is $2 - \frac{1}{m}$~\cite{RST21}. 
\begin{theorem}
\label{thm:poa_general}
    For $m\ge2$, $PoA(\CidenticalMachines^{-DA})= 3-\frac{1}{m}$.
\end{theorem}
\begin{proof}
Let $G \in \CidenticalMachines^{-DA}$. Let $\sigma^*$ be an optimal schedule of $G$. Denote by $OPT$ the makespan of $\sigma^*$. Let $\sigma$ be any $NE$ profile of $G$.
We first bound the total length of jobs in $\sigma$.

\begin{claim}
\label{ratioOPTMax}
  For every profile $\sigma$ of $G$, it holds that $\sum_i p_i^{\sigma} \le 2m \cdot OPT$.
\end{claim}
\begin{proof}
Consider the schedule $\sigma$. The start time of every job whose completion time is more than $OPT$ is clearly later than its start time in $\sigma^*$, therefore, for each such job $p_i^{\sigma} \le p_i^{\sigma^*}$. This implies that the total length of such jobs is at most $m \cdot OPT$. In addition, the total processing time of jobs that finish before $OPT$ is at most $m\cdot OPT$. We conclude that the total job length in $\sigma$ is at most $2m \cdot OPT$.
\end{proof}

By Observation~\ref{obs:neForIdenticalMachineIsLs} every $NE$, in particular $\sigma$, is a possible outcome of LS algorithm. 
Let $C_{max}(\sigma)$ denote the makespan of $\sigma$, and let $k$ be the job that finishes last and determines $C_{max}(\sigma)$.
We can now follow the analysis of LS algorithm to show
$    C_{max}(\sigma)\le \frac{1}{m} \sum_i p^{\sigma}_i + p^{\sigma}_k \cdot \frac{m-1}{m}. $

Now, consider the optimal schedule $\sigma^*$. The start time of job $k$ is clearly earlier in $\sigma^*$. Therefore, $p_k^{\sigma} \le p_k^{\sigma^*}$, implying that 
$OPT \ge p^{\sigma}_k$.
In addition, by Claim~\ref{ratioOPTMax}, we have that
$\sum_i p^{\sigma}_i \le 2m\cdot OPT$.
Combining these, we conclude
$C_{max}(\sigma)\le \frac{m-1}{m}\cdot OPT+ 2\cdot OPT = (3-\frac{1}{m})\cdot OPT.$

The above bound is tight. Specifically,  for any $0<\epsilon <1$ and $m \ge 2$, 
we define the following game, $G \in \CidenticalMachines^{-DA}$, played on $m$ identical machines, such that $PoA(G) = 3-\frac{1}{m}-\epsilon$. An example for $m=3$ is given in Figure~\ref{fig:PoA3m-1}.

Let $\epsilon'=\frac m{2m-1}\cdot \epsilon$. The set $N$ consists of three sets of jobs. The first set includes a single job, $a$, for which $p_a(t)=1-\epsilon'$. The second set is $U=\{u_1,\ldots,u_{(m-1)m}\}$. For every $u_i \in U$, let $p_{u_i}(t)=\frac{1-\epsilon'}m$. The third set is $V=\{v_1,\ldots,v_m\}$. For every $v_i \in V$, let $p_{v_i}(t) = \max\{1-t,\epsilon'\}$. Note that Job $a$ as well as the $U$-jobs have fixed processing times, while the $V$-jobs have decreasing processing times.

We turn to describe the machines' priority lists. For $1 \le j \le m$, let 
    $\pi_j=(v_j,U,a,V\setminus \{v_j\})$.  Whenever the list includes a set, the set elements are ordered according to their index.
Note that every machine gives preference to a different job from $V$.
    
In a possible optimal schedule, the first machine processes Job $a$ and Job $v_m$. For every $2 \le j \le m$, machine $j$ processes $m$ jobs from $U$ and Job $v_{j-1}$. 
Note that each of the $V$-jobs starts its processing at time $1-\epsilon'$. Therefore, the makespan of the optimal schedule is $1$. 
It easy to verify that this schedule is optimal, since it is balanced, and each job has its minimal possible length.

On the other hand, the following schedule, $\sigma$, is a $NE$: 
for every $1 \le j \le m$, machine $j$ processes job $v_j$, as well as the set $\{u_{mk +j}\}_{0 \le k \le m-2}$ which includes $m-1$ jobs from $U$. In addition, the first machine processes also Job $a$. 
The makespan of this schedule is obtained on the first machine and is equal to $1+(m-1)\cdot\frac{1-\epsilon'}m+1-\epsilon' = 3 - \epsilon'- \frac 1 m -\frac{m-1}m \cdot \epsilon' =  3 - \frac 1 m - \epsilon$.

We show that $\sigma$ is a $NE$. First, the $V$-jobs are all first, and achieve their minimal possible completion time, which is $1$. The $U$-jobs are assigned to the machines such that jobs with a low index appear before jobs with higher index, therefore, no $U$-job has a beneficial deviation. Finally, for all 
$j$, Job $a$ has the lowest priority among jobs that are assigned on machine $j$. 

We conclude that $PoA(G)=3-\frac 1 m -\epsilon$.

\end{proof} 
\begin{figure}[H]
\centering
\includegraphics[width=1\textwidth,trim=0 0 0 35,clip]{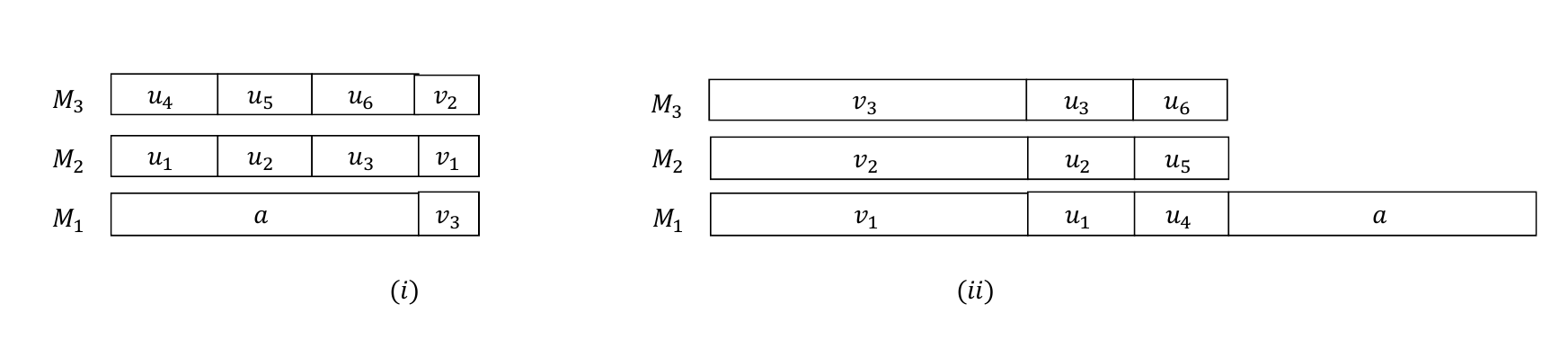}
\caption{A game on $3$ identical machines, with $PoA$ $3-1/3-\epsilon$. $(i)$ an optimal schedule with makespan $1$, $(ii)$ a NE schedule with makespan $3-1/3-\epsilon$.}
\label{fig:PoA3m-1}
\end{figure}

\subsubsection{Global Priority List} 

Let $\CidenticalGlobal^{-DA} =\CidenticalMachines^{-DA} \cap \Cglobal$. That is, the class of games with a global priority list, identical speeds, and non-decreasing completion time.

\begin{theorem}
    For $m\ge2, ~PoA(\CidenticalGlobal^{-DA})= 3-\frac{2}{m}$.
    \label{poa:3-2m} 
\end{theorem}
\noindent{\bf Proof Overview:} We first simplify the game by establishing that, w.l.o.g.,  jobs can be categorized into two disjoint sets: $F$, containing jobs with fixed lengths ($a_i=0$), and $D$, containing jobs that are longer in the Nash Equilibrium ($\sigma$) than in the social optimum ($\sigma^*$) due to earlier starting times. By assuming a worst-case deterioration rate of $a_i=1$ for jobs in $D$, we prove that each machine in a NE can process at most one job from $D$. This structural property allows for a tight bound on the total processing time in the NE, specifically $\sum_{i}p_{i}^{\sigma}\le(2m-1)\cdot OPT$. The final bound of $3-\frac{2}{m}$ is then derived by applying the standard analysis for the List-Scheduling algorithm.

The analysis is tight: we show that 
    for any $0<\epsilon <1$ and $m>1$, there exists a game in $\CidenticalGlobal^{-DA}$ played on $m$ identical machines, such that $PoA(G) = 3-\frac{2}{m}-\epsilon$. In this game, the global list is $(U_1,V,U\setminus U_1,a)$, where $U_1$ includes $m$ short fixed jobs, $V$ includes $m$ decreasing jobs, and $a$ is a long fixed job (see Figure~\ref{fig:global_3m-1}). In an optimal schedule, the $V$-jobs follow all of $U_1$ on one machine and shrink to their minimal length; in the NE, every $V$-job follows a single job of $U_1$ and remains long. The full construction is given in the appendix. \hfill \qed

\begin{figure}[h]
\centering
\includegraphics[width=0.92\textwidth]{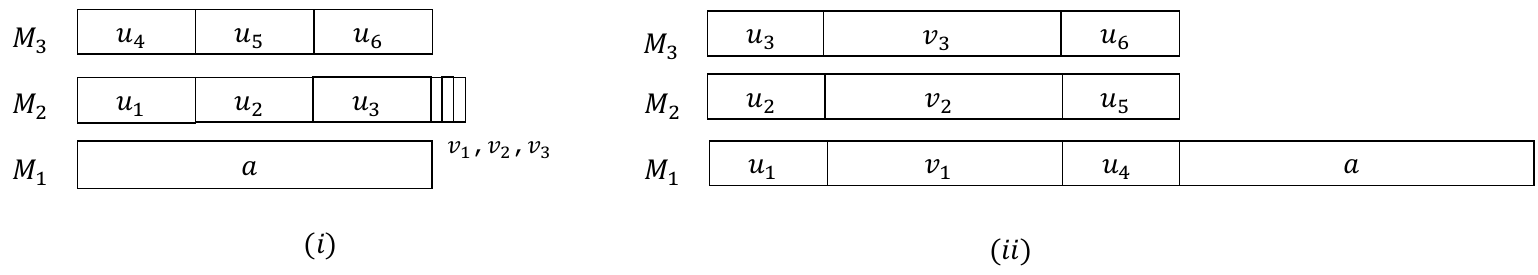}
\caption{A game with a global priority list $\pi=(u_1,u_2,u_3,v_1,v_2,v_3,u_4,u_5,u_6,a)$ and $3$ machines. $(i)$ an optimal schedule with makespan $1+O(\epsilon)$, $(ii)$ a $NE$ schedule with makespan $(3-\frac{2}{m}-O(\epsilon))$.}
\label{fig:global_3m-1}
\end{figure}

\section{Efficient Coordination Mechanisms}
\label{sec:efficient}

In this section we consider the problem of developing local scheduling policies that guarantee low equilibrium inefficiency of games in which jobs have time-dependent processing times. Overcoming this inefficiency is fundamentally why coordination mechanisms were introduced.

\subsection{Policies for Negative Deterioration Rates and Delay-Averse Jobs}
\label{sec:SDR}

In this section we consider games played on identical machines, all having speed $s=1$ and delay-averse jobs with negative deterioration rates, that is, the processing time function is of type $p_i(t)=max(b_i -a_i\cdot t ,\tau)$, where $\tau$ is fixed for all jobs, and $0\le a_i\le 1$. Recall that delay-averse jobs have an incentive to start their processing as early as possible.

In Theorem~\ref{poa:3-2m} we showed that the $PoA$ for games with global priority list and non-increasing completion time is $3-\frac{2}{m}$. We present two mechanisms with a lower performance guarantee. 

\subsubsection{SDR Scheduling Policy}
The \emph{Smallest Deterioration Rate} (SDR) policy sorts the jobs in a non-decreasing order of their negative deterioration rate, $(a_1\le a_2\le \ldots \le a_n)$. We prove that with this natural policy the $PoA$ reduces to $2$ independent of the number of machines. Our proof introduces a new technique which characterize for each job $k$ a value $W_k$, representing the total work done by all machines until the moment it begins. By focusing on this workload characterization, instead of the simple starting time of the job, it is possible to show that any profile that follows the $SDR$ policy maximizes the total deterioration of unassigned jobs, for every $W_k$. Using this observation, we show that $SDR$ guarantees that the starting time of the last job in a $NE$ profile is at most the makespan of an optimal profile. 
Let $G \in \CidenticalGlobal^{-DA}$ and let $\pi_{SDR}=(1,\dots,n)$ be the priority list corresponding to SDR policy, that is,  $a_1 \le \ldots \le a_n$.

\begin{theorem}
\label{th:sdr2}
If game $G$ is played with $\pi_{SDR}$, then $PoA(G) = 2$.
\end{theorem}
\begin{proof}
For any priority list $\pi$, let $\sigma_\pi$ be a profile generated by running List-Scheduling on $\pi$. In $\sigma_\pi$, all $m$ machines are busy as long as some job is not assigned. Let $\sigma_{SDR}$ be a profile generated by $\pi_{SDR}$. W.l.o.g., let $n$ be the last job to start in $\sigma_{SDR}$. We assume that no job reaches the threshold processing time $\tau$ (We show how to relax this assumption in the appendix).
Since $\sigma_{SDR}$ follows LS, all $m$ machines are busy during $[0, S_n^{\sigma_{SDR}}]$.

For any job $k \in \{1, \dots, n\}$, let $W_k(\sigma)$ be the total work done by the $m$ machines in profile $\sigma$ until the moment job $k$ begins its execution. In a system without idle time, $W_k(\sigma) = m \cdot S_k^\sigma$.
Note that for $k \le m$, $W_k(\sigma) = 0$ (as the first $m$ jobs start at $t=0$). 
For any schedule $\sigma$ and workload $w$, let $B_{\sigma}(w)$ be the set of jobs whose processing has not yet begun at cumulative workload $w$, that is $i \in B_{\sigma}(w)$ iff $W_i(\sigma) >w$.

Next, we define the sum of deterioration rates of jobs that have not started their processing, as a function of $w$, the work already done. 
The {\em work-based deterioration density} is defined as follows for all $w \ge 0$:
$\mathcal{A}_\sigma(w) =
\sum_{i: W_i(\sigma) > w} a_i = \sum_{i \in B_{\sigma}(w)} a_i.$

This is a non-increasing step function where the "stairs" occur at the points $W_k(\sigma)$.

For example, at time $t=0$ the first $m$ jobs starts their processing. When one of them completes, job $m+1$ starts its processing. For $0<w \le W_{m+1}(\sigma)$, we have that $\mathcal{A}_\sigma(w) = \sum_{i: W_i(\sigma) > w} a_i= \sum_{i=m+1}^n a_i$.
After job $m+1$ starts its processing, and till job $m+2$ starts its processing, that is, for  for $ W_{m+1}(\sigma) <w \le  W_{m+2}(\sigma)$,  $\mathcal{A}_\sigma(w) = \sum_{i: W_i(\sigma) > w} a_i= \sum_{i=m+2}^n a_i$. Indeed, during this work interval, every job $i \in \{m+2,\ldots,n\}$ is shortened at rate $a_i$.

\begin{claim}
\label{cl:SDRdensity}
For any $0 \le w \le W_n(\sigma_{SDR})$ the SDR policy maximizes the cumulative deterioration density $\int_0^w \mathcal{A}_{\sigma}(u) du$, among all profiles with no intended idles.
\hfill \qed
\end{claim}

Let $\sigma^*$ be an optimal profile, and let $OPT = C_{max}(\sigma^*)$. The next claim provides a bound of OPT on the starting time of job $n$ in $\sigma_{SDR}$. 

\begin{claim}
\label{cl:S_nleOPT}
    $S_n^{\sigma_{SDR}} \le OPT$. \hfill \qed
\end{claim}

The makespan of the NE produced by SDR is $C_{max}(\sigma) = S_n^{\sigma_{SDR}} + p_n(S_n^{\sigma_{SDR}}) = S_n^{\sigma_{SDR}} + (b_n - a_n S_n^{\sigma_{SDR}}) = b_n + S_n^{\sigma_{SDR}}(1 - a_n)$. Recall that  $a_n \in [0, 1]$ and job $n$ is delay-averse, thus, $b_n \le OPT$. By Claim~\ref{cl:S_nleOPT}, $S_n^{\sigma_{SDR}} \le OPT$. We conclude $C_{max}(\sigma) \le OPT + OPT(1 - a_n) \le 2 \cdot OPT$. Thus, $PoA(G) = \frac{C_{max}(\sigma)}{OPT} \le 2.$

In the appendix we show that the analysis is tight. Specifically,
    For any $\delta>0$, and $m>1$, there exists a game $G^{SDR}$ on $m$ machines with SDR policy, such that $PoA(G^{SDR})\ge2-\delta$. 
\end{proof}

\subsubsection{LBDR Scheduling Policy}

The \emph{Largest Basic-Deterioration Ratio} (LBDR) policy sorts the jobs in non-increasing order of $\frac{b_i}{a_i}$. We assume w.l.o.g., that for all $i$, $a_i>0$. If a job has $a_i=0$, then we set $a_i =\epsilon$ for a very small $\epsilon$. Note that such jobs will appear first in the resulting priority list, ordered by their basic processing times. In addition, we assume that the job that determines $C_{max}(\sigma)$ does not reach $\tau$, and we conjecture the bound extends to the general case (see appendix).

\begin{theorem}
        \label{th:LBDR}
       $PoA(G^{LBDR})=\frac{e}{e-1}$ for $m\le3$, and $2-\frac{1}{m}$ for $m>3$.
\end{theorem}
\noindent{\bf Proof Overview:} Consider a game $G^{LBDR}\in\CidenticalGlobal^{-DA}$ on $m$ identical machines, and let $R_i=\frac{b_i}{a_i}$. Our proof fully described in the appendix, is structured as follows.  Let $\sigma^*$ be an optimal profile with makespan $OPT$, and let $\sigma$ be a Nash equilibrium. W.l.o.g., we assume that job $n$ determines $C_{\max}(\sigma)$ and starts at time $S_n$, and that it does not reach its threshold processing time $\tau$.

We first derive that $C_{\max}(\sigma)=b_n+S_n(1-a_n)<R_n$, and focus on the case where $C_{\max}(\sigma)>OPT$, which implies $R_n>OPT$. Using the ordering induced by $LBDR$, we bound the processing times of earlier jobs and obtain, for every job $i$, a relation between $R_n-C_i$ and $b_i$. This yields a multiplicative bound on completion times along any machine in $\sigma^*$, which implies
$\left(1-\frac{OPT}{R_n}\right)^m \le \prod_{i=1}^n\left(1-\frac{b_i}{R_n}\right)$.
We then use the structure of $\sigma$ as a List-Scheduling algorithm outcome to show that $\sum_{i=1}^{n-1} b_i \ge m\cdot S_n$, and combine this with the inequality $1-u\le e^{-u}$ to upper bound the product over the first $n-1$ jobs. This yields an upper bound on $S_n$ in terms of $R_n$, $OPT$, and $b_n$.

Substituting this bound into $C_{\max}(\sigma)=b_n+S_n(1-a_n)$ gives an upper bound on the makespan as a function of $b_n$. We show that this function is strictly convex over $b_n\in[0,OPT]$, and therefore its maximum is attained at the boundary points. Evaluating these cases yields $C_{\max}(\sigma)\le (2-\frac{1}{m})\cdot OPT$ for $b_n=OPT$, and $C_{\max}(\sigma)\le \frac{e}{e-1}\cdot OPT$ as $b_n\to 0$.
\hfill\qed

\subsection{Policies for Uniform Positive Deteriorating Rate}
\label{sec:SBPT}

Our next policy is for jobs having arbitrary basic processing times and a uniform deterioration rate, $a>0$. Formally, $p_i(t)=b_i +a\cdot t$, where $b_i\ge0$. In Theorem~\ref{thm:UniformRateGlobalCombined} we showed that with an arbitrary global policy, the $PoA$ can be as high as $(1+a)^{\frac{n}{m}}$.
To reduce the $PoA$ we apply the policy \emph{Shortest Basic Processing Time} (\emph{SBPT}), that schedules the jobs in non-decreasing order of their basic processing times $(b_1\le b_2 \le \ldots \le b_n)$, with arbitrary tie-breaking.

\emph{SBPT} policy is inspired by~\cite{BrowneYechieaRandom,GuptaGuptaSingel} who show that on a single machine, scheduling the jobs in a non-decreasing order of $\frac{b_i}{a_i}$ minimizes the makespan.
The central argument of our proof establishes that a NE profile following the \emph{SBPT} rule minimizes the total processing time. We first characterize the structure of any sum-minimizing schedule, showing it consists of $z$ rounds of $m$ jobs based on the global priority list. We then prove that the profile generated by List-Scheduling with \emph{SBPT} maintains this same structure. Since the total processing time of the stable profile is minimal, we can more effectively bound and relate its makespan to that of an optimal schedule.

Let ${\mathcal G}^{SBPT} \subset \CidenticalGlobalUniformRate$ be the class of games with \emph{SBPT} policy. We provide the folowing upper bound on the PoA. A tight example for  $m=2$ and $a=1$ is given in the appendix.
\begin{theorem}
    \label{thm:SBPTCombined}
     $PoA({\mathcal G}^{SBPT}) \leq \frac{2m+a\cdot m-1}{m+a} .$ \hfill \qed
\end{theorem}

\section{Conclusions and Directions for Future Work}
\label{sec:conclude}

We studied how coordination mechanisms affect strategic scheduling when jobs have time-dependent processing times. We provide a complete picture of this setting, distinguishing between various classes of instances. Our work bridges the gap between two well-explored areas: job scheduling games, in particular coordination mechanisms, and the centralized optimization problem of scheduling jobs with time-dependent processing times.

We introduce \emph{delay-averse} jobs, whose completion time function is non-decreasing, and identified several classes of games where $NE$ is guaranteed, with List-Scheduling and best-response dynamics (BRD) converging to equilibrium. We analyzed the effect of having deteriorating jobs on the equilibrium inefficiency, and showed that the $PoA$ with respect to the makespan can grow exponentially in the number of jobs, even with uniform positive deterioration rates and a uniform machine policy.
For negative deterioration, $NE$ existence depends mainly on the relation between job rates and machine speeds; delay-averse jobs admit $NE$ in an environment of identical machines. Equilibria can be less efficient than in fixed processing times, although the $PoA$ remains constant.

To reduce inefficiency, we propose three coordination mechanisms: (1) SBPT (Shortest Basic Processing Time) for positive deterioration, guaranteeing constant $PoA$ versus the exponential worst case under arbitrary policies; (2) SDR, a mechanism ordering jobs by non-decreasing deterioration rate for negative deterioration, achieving $PoA = 2$ independent of machine count, with a lower bound of $2-\delta$; and (3) LBDR, ordering jobs by non-increasing $\frac{b_i}{a_i}$, reducing $PoA$ to $\frac{e}{e-1}$ for $m\le 3$ and $2-\frac{1}{m}$ for $m\ge 4$.

Several directions remain open. Deterioration could instead be modeled as machine-dependent rather than job-dependent, capturing machines that degrade at different rates. Most of our negative-deterioration results assume identical machine speeds; relaxing this is likely to yield richer, less predictable behavior, particularly for non-delay-averse jobs, and is practically motivated since real systems rarely have identical resources. The sensitivity of equilibrium structure to the minimal processing-time threshold $\tau$, especially if it varies by job or machine, also warrants further study. Furthermore, the concept of having such a threshold has not yet been explored in classical, non-strategic scheduling environments. 
Another direction is to study mixed-policy environments, where, for example, some machines apply SDR and some LBDR, 
may reveal new equilibrium phenomena and inefficiency bounds. Finally, this work focused solely on makespan; extending the analysis to objectives such as total completion time is a natural next step.

\appendix
\section{Missing Proofs}

\noindent{\bf Proof of Observation~\ref{ob:optimalStartTime}:}
With a negative deterioration rate, where $p_i(t) = \max(\tau_i,b_i-a_i\cdot t)$, we have that 
$$c_i(t,s)= t + \max(\frac{\tau_i}{s},\frac{b_i}{s}-\frac{a_i}{s}\cdot t)=\max(t+\frac{\tau_i}{s},\frac{b_i}{s}+t(1-\frac{a_i}{s})).$$

The first term is clearly minimal for $t=0$. We distinguish between three cases.
\begin{enumerate}
    \item If $a_i<s$, then $(1-\frac{a_i}{s})>0$, and the second term increases as $t$ increases, and is minimal for $t=0$.
    \item If $a_i>s$, then $(1-\frac{a_i}{s})<0$, and the second term reduces as $t$ increases. Thus, $c_i(t,s)$ is minimal when $t+\frac{\tau_i}{s}=\frac{b_i}{s} + t(1-\frac{a_i}{s})$, that is, when $\frac{\tau_i}{s}=\frac{b_i}{s}-\frac{a_i\cdot t}{s}$, which holds for $t=\frac{b_i-\tau_i}{a_i}$.
    \item If $a_i=s$, then the completion time remains constant for $t\le \frac{b_i-\tau_i}{a_i}$ and is given by $\frac{b_i}{s}$. Thus, $c_i(t,s)$ is minimal when $t\le \frac{b_i-\tau_i}{a_i}$.
\end{enumerate}
\hfill\qed

\posA

\noindent{\bf Proof of Theorem~\ref{thm:unitjobs}:}
The following algorithm assigns the jobs greedily, where in each step, a job is added on a machine on which the cost of a next job is minimized. This algorithm was shown to produce a $NE$ in games without time-dependent processing times, that is, for $a_i=0$ \cite{RST21}. We show that it produces a $NE$ also for instances in $\CidenticalJobs$.

\begin{algorithm}[H]
\caption{ Calculating a $NE$ of symmetric jobs on related machines} \label{NEUWjobs}
\begin{algorithmic}[1]
\STATE Let $\ell_j$ denote the completion time of machine $j$.  Initially, $\ell_j=0$ for all $1 \le j \le m$.\REPEAT 
\STATE Let $j^{\star} = \argmin_j ~\ell_j+p(\ell_j)/s_j$.
\STATE Assign on machine $j^{\star}$ the first unassigned job on its priority list.
\STATE $\ell_{j^{\star}} = \ell_{j^{\star}}+p(\ell_{j^{\star}})/s_{j^{\star}}$.
\UNTIL {all jobs are scheduled}
\end{algorithmic}
\end{algorithm}

 Let $\sigma^{\star}$ be the schedule produced by Algorithm \ref{NEUWjobs}. We show that $\sigma^{\star}$ is a $NE$.   Note that the jobs are assigned one after the other according to their completion time in $\sigma^{\star}$. That is, if $u$ is assigned before $v$ then $C_{u}(\sigma^{\star}) \le C_{v}(\sigma^{\star})$. 
 Assume by contradiction that $\sigma^{\star}$ is not a $NE$, and let $i$ be a job that can migrate from its current machine $j$ to machine $j'$ and reduce its completion time. Assume job $i$ is assigned as the $k$-th job on its current machine, and that if it migrates, then $i$ would be assigned as the $k'$-th job on machine $j'$.  Note that before the migration, job $i$ has the same completion time as the $k$-th job on machine $j'$, so $k'<k$. However, this contradicts the choice of the algorithm when the $k'$-th job on machine $j'$ is assigned - since $i$ should have been selected. If no job is $k'$-th on machine $j'$, then we get a contradiction to the assignment of $i$.

\hfill\qed

\posA

\noindent{\bf Proof of Theorem~\ref{thm:m2NE}:}
Assume, w.l.o.g., that $s_1=1$ and $s_2=s \le 1$.
The following algorithm computes a $NE$ for instances in $\Ctwo^{+}$ and in $\Ctwo^{-DA}$. It initially assigns all the jobs on the fast machine. Then, the jobs are considered according to their order in $\pi_2$, and every job gets an opportunity to migrate to $M_2$.
\begin{algorithm}[H]
\caption{ Calculating a $NE$ schedule on two related machines} \label{NE2machines}
\begin{algorithmic}[1]
\STATE Assign all the jobs on $M_1$ (fast machine) according to their order in $\pi_1$.
\STATE For $1 \le k \le n$, let the job $i$ for which $\pi_2(i)=k$ perform a best-response move (migrate to $M_2$ if this reduces its completion time).
\end{algorithmic}
\end{algorithm}

Let $\sigma$ denote the schedule after the algorithm terminates. The following two claims show that after the termination of the algorithm, no job has a unilateral deviation that improves its cost, i.e., $\sigma$ is a $NE$.

For simplicity, we assume that the deterioration rates of the jobs are positive, as it does not change the validity of the claims.

\begin{claim}
\label{cl:dontleaveM1}
No job for which $\sigma_i=M_1$ has a beneficial migration.
\end{claim}
\begin{proof}
Assume by contradiction that job $i$ is assigned on $M_1$ and has a beneficial migration. Assume that $\pi_2(i)=k$. Job $i$ was offered to perform a migration in the $k$-th iteration of step 2 of the algorithm, but chose to remain on $M_1$. The only migrations that took place after the $k$-th iteration are from $M_1$ to $M_2$. Thus, if migrating is beneficial for $i$ after the algorithm completes, it should have been beneficial also during the algorithm, since its completion time could only decrease, contradicting its choice to remain on $M_1$.
\end{proof}

\begin{claim}
\label{cl:dontleaveM2}
No job for which $\sigma_i=M_2$ has a beneficial migration.
\end{claim}
\begin{proof}
Assume by contradiction that the claim is false and let $i$ be the first job on $M_2$ (first with respect to $\pi_2$) that may benefit from returning to $M_1$.
Let $\widehat{\sigma}$ denote the schedule before job $i$ migrates to $M_2$ - during the second step of the algorithm.
Recall that $C_i(\sigma)$ is the completion time of job $i$ on $M_2$, and $C_i(\widehat{\sigma})$ is its completion time on $M_1$ before its migration.

Since the jobs are activated according to $\pi_2$ in the $2$-nd step of the algorithm, no jobs are added before job $i$ on $M_2$. Job $i$ may be interested in returning to $M_1$ only if some jobs that were processed before it on $M_1$, move to $M_2$ after its migration. Denote by $\Delta$ the set of these jobs, and let $\delta$ be their total actual running time. Let $i'$ be the last job from $\Delta$ to complete its processing in $\sigma$. Since job $i'$ performs its migration out of $M_1$ after job $i$, and jobs do not join $M_1$ during step 2 of the algorithm, the completion time of $i'$ when it performs the migration is at most $C_{i'}(\widehat{\sigma})$. Denote by $\Gamma$ the set of jobs that were processed before $i$ in $M_1$ and are not part of $\Delta$, and denote by $\gamma$ their actual running time. By the definition of $p$ we have that $C_i(\widehat{\sigma}) = \gamma + \delta + p(\delta + \gamma) = \gamma + \delta + b_i + a_i\cdot (\delta + \gamma)$. Since $i'$ is the last to leave $M_1$ in $\Delta$, it leaves $M_1$ when the actual running time of jobs before it is $\beta \le \gamma$. By the definition of $p(i')$, $C_{i'}(\widehat{\sigma}) = \beta + a_{i'}\cdot \beta +b_{i'}$.
The migration from $M_1$ to $M_2$ is beneficial for $i'$, thus, $C_{i'}(\sigma) < C_{i'}(\widehat{\sigma})$. 

The jobs in $\Delta$ are all before job $i$ in $\pi_1$ and after job $i$ in $\pi_2$. Specifically job $i'$ is processed after $i$, thus, its starting time is at least $C_i(\sigma)$. Since $i'$ is delay averse, its completion time higher, and is at least $C_{i'}(\sigma) \ge C_{i}(\sigma)+(a_{i'}\cdot C_{i}(\sigma)+b_{i'})/s $. 

Finally, we assume that $\sigma$ is not stable and job $i$ would like to return to $M_1$. Let $\gamma'$ be the actual running time of jobs processed on $M_1$ before job $i$ when it returns from $M_2$. Note that $\gamma \ge \gamma' \ge \beta$ and that the completion time of $i$ is $\gamma' + a_i\cdot \gamma' + b_i$. The migration of job $i$ to $M_1$ implies that $\gamma' + a_i\cdot \gamma' + b_i< C_{i}(\sigma)$.

Combining the above inequalities, we get 
\begin{align*}
(\gamma' + a_i\cdot \gamma' + b_i)+(\gamma' + a_i\cdot \gamma' + b_i+b_{i'})/s <  \\ C_{i}(\sigma) + (a_{i'}\cdot C_{i}(\sigma)+b_{i'})/s < C_{i'}(\sigma) 
 &<C_{i'}(\widehat{\sigma}) \le \beta + a_{i'}\cdot \beta + b_{i'} < \gamma' + a_{i'}\cdot \gamma' +b_{i'},
\end{align*}
and specifically $(\gamma' + a_i\cdot \gamma' + b_i)+(\gamma' + a_i\cdot \gamma' + b_i+b_{i'})/s<\gamma' + a_{i'}\cdot \gamma' +b_{i'} $ which is a contradiction since all values are non negative and $s \le 1$.

Note that if the jobs have negative deterioration rates, then $C_{i'}(\widehat{\sigma}) = \beta + \max\{\tau_{i'}, b_{i'}-a_{i'}\cdot \beta\}$. Since $i'$ is delay-averse and $\beta \le \gamma'$, we have $C_{i'}(\widehat{\sigma}) \le \gamma' + \max\{\tau_{i'}, b_{i'}-a_{i'}\cdot \gamma'\}$. 
On the other hand, $i'$ starts on $M_2$ after $C_{i}(\sigma)>\gamma'$. Since $i'$ is delay-averse and $s \le 1$, we get $C_{i'}(\sigma) \ge \gamma' + \max\{\tau_{i'}, b_{i'}-a_{i'}\cdot \gamma'\}/s \ge C_{i'}(\widehat{\sigma})$, a contradiction.
\end{proof}

\hfill\qed

\posA

\noindent{\bf Proof of Theorem~\ref{th:brdcycle}:}
We show that every sequence of best-response converges to a $NE$ for every instance  $G \in \CidenticalJobs \cup \Ctwo^+ \cup \CidenticalMachines^{DA} \cup \Cglobal$. 

In the following proofs we extend the proof in~\cite{RST21} for jobs with fixed processing times. The proofs have the same structure for all classes. We assume that best-response dynamics (BRD) does not converge. Since the number of different profiles is finite, the sequence of profiles contains a loop. That is, the sequence includes a profile $\sigma_0$, starting from which jobs migrate and eventually return to their strategy in $\sigma_0$. Let $\Gamma$ denote the set of jobs that perform a migration during this loop. For each of the four classes we identify a job $i \in \Gamma$ such that once job $i$ migrates, it cannot have an additional beneficial move.

Consider first a game $G \in \CidenticalJobs$, that is, a game in which all jobs have the same processing time function $p_i$.
Let $C_{min}$ be the lowest completion time of a job in $\Gamma$ during the BR-cycle. Let $M_1$ be a machine on which $C_{min}$ is achieved. Let $i$ be the job achieving completion time $C_{min}$ on $M_1$ with the highest priority on $M_1$ among the jobs in $\Gamma$ that achieve cost $C_{min}$ during the loop. Once $i$ achieves completion time $C_{min}$, no job is added to $M_1$ before it. Job $i$ cannot have an additional beneficial move, as this will contradict the definition of $C_{min}$.

We turn to consider games in $\Ctwo^+$, that is, $G$ is played on two machines. W.l.o.g., assume $s_1=1$ and $s_2=s \le 1$.  Let $i$ be the job in $\Gamma$ with highest priority in $\pi_2$. 
Given that BRD loops and that $i \in \Gamma$, it holds that during the BR sequence $i$ migrates from $M_1$ to $M_2$ and then back from $M_2$ to $M_1$. 

We show that once $i$ moves from $M_1$ to $M_2$, moving back to $M_1$ cannot be beneficial for it.
Let $\sigma^\prime$ denote the schedule before job $i$ migrates from $M_1$ to $M_2$.
Assume by contradiction that $i$ may benefit from returning to $M_1$. 
Let $L_1$ be the total processing time of jobs on $M_1$ that precede $i$ on $\pi_1$ in $\sigma^{\prime}$. We have that $C_i(\sigma^{\prime})=L_1+p_i(L_1)$.
Let $L_2$ be the completion time of the jobs in $N \setminus \Gamma$ that precede $i$ on $\pi_2$.
Since $i$ has the highest priority among $\Gamma$ on $M_2$, its cost while on $M_2$ is $L_2+p_i(L_2)/s$, independent of other jobs leaving and joining $M_2$.
The migration of $i$ from $M_1$ to $M_2$ is beneficial, thus, $L_1+p_i(L_1) > L_2+p_i(L_2)/s$.
Migrating back to $M_1$ may become beneficial only if the total processing time of job that would precede it on $M_1$ is less than $L_1$, thus, at least one job that precedes $i$ on $\pi_1$ migrates out of $M_1$ when $i$ is on $M_2$. Let $k$ be the last job, for which $\pi_1(k) < \pi_1(i)$ that have left $M_1$ when $i$ is on $M_2$. Following $k$'s migration the processing time of jobs on $M_1$ that precede $i$ in $\pi_1$ is $L'_1$. Migrating back is beneficial for $i$, thus, $L'_1+p_i(L'_1) <L_2+p_i(L_2)/s$ (additional jobs may join $M_1$ after $k$ leaves it, but this only makes $M_1$ less attractive for $i$). 
Since $\pi_2(k) > \pi_2(i)$, the cost of $k$ after its migrating to $M_2$ is at least $L_2+p_i(L_2)/s+p_k(L_2+p_i(L_2)/s)/s$. $k$'s migration from $M_1$ to $M_2$ is beneficial, thus, $L'_1+p_k(L'_1) > L_2+p_i(L_2)/s+p_k(L_2+p_i(L_2)/s)/s$. 
By combining the above inequalities we reach a contradiction. Specifically, let $X=L_2+p_i(L_2)/s$
$L'_1+p_i(L'_1) < X=(X+p_k(X)/s)-p_k(X)/s < L'_1+p_k(L'_1) -p_k(X)/s \le L'_1$.
We conclude that job $i$ cannot benefit from returning to $M_1$ and thus, cannot be involved in the BRD-cycle.

Assume next that $G \in \CidenticalMachines^+ \cup \CidenticalMachines^{-DA}$, that is, machines have identical speeds and jobs are delay-averse. Let $t$ be the lowest start time of a job in $\Gamma$ during the BR-cycle. Let $M_1$ be a machine on which $t$ is achieved. Let $i$ be the job in $\Gamma$ with highest priority on $M_1$. Clearly, once $i$ achieves start time, $t$, it cannot have an additional beneficial move, as this will contradict its choice.

Recall that we previously showed that a $NE$ is not guaranteed in classes $\G{^-_2} \cup  \G{^-_3}$, and hence BRD does not converge to $NE$.

Finally, if $G \in \Cglobal$, that is, when machines share a global priority list, then once the job in $\Gamma$ with the highest priority migrates, it selects the machine with the lowest total processing time of jobs in $N \setminus \Gamma$ that precedes it, and cannot have an additional beneficial move later.
\hfill \qed
\posA

\noindent{\bf Proof of Theorem~\ref{thm:dwin2M2}:}
Let $G$ be a game with two identical machines $M=\{M_1,M_2\}$ and two jobs $N=\{u,v\}$. Assume w.l.o.g., that $\pi_1 = (u,v)$. Consider the profile in which job $u$ is assigned on $M_1$ and job $v$ is assigned on $M_2$, and assume that job $u$ has a beneficial deviation to $M_2$. Such a migration may change $C_u$ only if job $u$ is assigned after job $v$, that is, $\pi_2=(v,u)$. Note that after the migration job $v$ has no incentive to migrate to $M_1$, since its cost will remain the same. Thus, this profile is a $NE$.
\hfill \qed
\posA

\noindent{\bf Proof of Theorem~\ref{thm:dwin2M3}:}
For $n=3$ below is a description of the possible profiles of the game $G_{noNE3}$. See Figure~\ref{fig:noNE3Jobs}.
\begin{enumerate}
    \item The jobs are all on $M_1$. $cost(u) = 4.8$, $cost(v) = 4$, $cost(w) = 5$.
    Job $w$ will benefit from migrating to $M_2$, reducing its cost to $3$.
    \item Jobs $u$ and $v$ are on $M_1$. $cost(u) = 4.8$, $cost(v) = 4$, $cost(w) = 3$.
    Job $v$ will benefit from migrating to $M_2$, reducing its cost to $3.7$.
    \item Job $u$ is on $M_1$. $cost(u) = 5$, $cost(v) = 3.7$, $cost(w) = 3$.
    Job $u$ will benefit from migrating to $M_2$, reducing its cost to $4.85$.
    \item The jobs are all on $M_2$. $cost(u) = 4.85$, $cost(v) = 5.05$, $cost(w) = 3$.
    Job $v$ will benefit from migrating to $M_1$, reducing its cost to $4$.
    \item Jobs $v$ and $w$ are on $M_1$. $cost(u) = 5$, $cost(v) = 5$, $cost(w) = 4.2$.
    Job $u$ will benefit from migrating to $M_1$, reducing its cost to $4.8$.
    \item Job $w$ is on $M_1$. $cost(u) = 5$, $cost(v) = 5.2$, $cost(w) = 3$.
    Job $v$ will benefit from migrating to $M_1$, reducing its cost to $4$.
    \item Jobs $u$ and $w$ are on $M_1$. $cost(u) = 5$, $cost(v) = 5$, $cost(w) = 5.2$.
    Job $w$ will benefit from migrating to $M_2$, reducing its cost to $3$.
    \item Job $v$ is on $M_1$. $cost(u) = 4.85$, $cost(v) = 4$, $cost(w) = 3$.
    Job $u$ will benefit from migrating to $M_1$, reducing its cost to $4.8$.
\end{enumerate}  
Thus, the game $G_{noNE3}$ has no pure Nash equilibrium.
\hfill\qed

\posA

\noindent{\bf Proof of Theorem~\ref{th:neISNP}:}
   Given a game and a profile $\sigma$, verifying whether $\sigma$ is a $NE$ can be done by checking for every job whether its current assignment is also its best-response. Therefore, the problem is in NP.
   
    The hardness proof is by a reduction from $3$-bounded $3$-dimensional matching ($3$DM-$3$).  The input to the $3$DM-$3$ problem is a set of triplets $T\subseteq X \times Y \times Z$, where $|T|\ge n$, and $|X|=|Y|=|Z|$. The number of occurrences of every element of $X\cup Y \cup Z$ in $T$ is at most $3$. The goal is to decide whether $T$ has a 3-dimensional matching of size $n$, i.e., there exists a subset $T'\subseteq T$ such that $|T|=n$ and every element in $X\cup Y \cup Z$ appears exactly once in $T'$. 
    $3$DM-$3$ is known to be NP-hard \cite{3DM}.
    Given an instance $T$ of 3DM-3 matching, we construct the following scheduling game $G_T$. Let $\Gamma=\frac{5-0.2}{1.05}\approx4.571$, $\tau=0.2$.
    The set of jobs consists of:
    
    \begin{enumerate}
        \item The 3 jobs $\{u,v,w\}$ from the game $G_{noNE3}$ in Theorem \ref{thm:dwin2M3}. Recall that their processing time functions are: $p_u(t)=max\{5-1.05\cdot t,0.2\}$, $p_v(t)=max\{4-1.1\cdot t,0.2\}$, $p_w(t)=max\{3-1.2\cdot t,0.2\}$.
        \item A single dummy job $f$, with processing time function $p_f(t)=\Gamma$ ($a_f=0$).
        \item A set $D$ of $|T|-n$ dummy jobs with processing time function $p_i(t)=\Gamma +1$
        \item A set $U$ of $|T|+3$ dummy jobs with processing time function  $p_i(t)=\Gamma+2$.
        \item $3n$ element-jobs with processing time $p_i(t)=1$, one for each element in $X \cup Y \cup Z$
    \end{enumerate}

     Note that all jobs except for $u,v,w$ have fixed processing time, and therefore prefer to start being processed as early as possible. There are $m=|T|+3$ identical machines, $M_1,M_2,\ldots,M_{|T|+3}$, with speed $s=1$.

 We first highlight some properties of the game $G_{noNE3}$.
   By Observation~\ref{ob:optimalStartTime}, and since all three jobs have negative deterioration rate more than $1$, the optimal starting times of the jobs are given by $min_t(i)=\frac{b_i-0.2}{a_i}$.
    Hence, $min_t(u)=\frac{5-0.2}{1.05}\approx4.571$,  $min_t(v)=\frac{4-0.2}{1.1}\approx3.45$, and $min_t(w)=\frac{3-0.2}{1.2}\approx2.3$.
    Next, note that if job $u$ is missing from $G_{noNE3}$ then there exists a $NE$ of $\{v,w\}$ on $M_1$ and $M_2$, when each of these jobs is assigned alone on the machine that prioritize the other job. Similar to~\cite{RST21}, the idea in our reduction is that if a $3$DM-$3$-matching exists, then the element jobs will be assigned on triplet-machines and job $u$ would prefer $M_3$. Thus, the game $G_{noNE3}$ will not be induced on machines $M_1$ and $M_2$ and a $NE$ exists. On the other hand, if a $3$DM-$3$-matching does not exist, then some element job will be assigned on $M_3$, making this machine not attractive for $v,u,w$. These three jobs will go to $M_1$ and $M_2$ and will bring to life the game $G_{noNE3}$.
  
    In order to achieve this particular behavior, we design the priority lists as follows. Whenever a priority list includes a set, the jobs in the set appear in arbitrary order. For the first machine, $\pi_1=(v,u,w,f,U,X,Y,Z,D)$. For the second machine, $\pi_2 =(w,u,v,f,U,X,Y,Z,D)$. Note that the prefix of the priority lists of the first two machines is the same as in $G_{noNE3}$.
    For the third machine, let $\pi_3=(f,X,Y,Z,D,u,v,w,U)$. 
    
    The remaining $|T|$ machines are {\em triplet-machines}. For every triplet $t=(x_i,y_j,z_k)\in T$, the priority list of the triplet-machine corresponding to $t$ is $(D,x_i.y_j,z_k,f,U,X\setminus\{x_i\},Y\setminus\{y_j\},Z\setminus\{z_k\},u,v,w)$. 
    
    Since $f$ is first in $\pi_3$, in any $NE$, the dummy job $f$ is assigned as the first job on $M_3$. In addition, there are $|D|=|T|-n$ triplet-machines on which the first is from $D$, since the jobs in $D$ have the highest priority on the triplet-machines.

    We show that $G_T$ has a $NE$ if and only if $T$ has a perfect matching. Figure~\ref{fig:NEReduction} illustrates the proof for $n=2$ and $|T|=3$.

    \begin{claim}
        \label{3dmToNE}
        If a 3D-matching of size $n$ exists in $T$, then $G_T$ has a $NE$.
    \end{claim}
    \begin{proof}
        Let $T'\subseteq T$ be a matching of size $n$. Assign the jobs of $X \cup Y \cup Z$ on the triplet-machines corresponding to $T'$, and the jobs of $D$ on the remaining triplet-machines. Assign $f$ and $u$ on $M_3$, and assign one job from $U$ on every machine.
        Jobs $v$ and $w$ are assigned on $M_1$ and $M_2$: job $w$ is on $M_1$ and job $v$ is on $M_2$. Note that the completion time of $w$ is $3$, the completion time of $v$ is $4$, and the completion time of $u$ is $\Gamma +\tau \ge 4.571+0.2$. The resulting schedule is indeed a $NE$, since all the jobs from $X \cup Y \cup Z$ complete at time at most $3$ and hence have no incentive to deviate to a different machine, and job $u$ starts in its optimal starting time $\Gamma$ on $M_3$.
    \end{proof}

    Now we show that if $T$ does not have a perfect matching, then the no-$NE$ game, $G_{noNE3}$, comes to life, and as a result $G_T$ has no $NE$.

    \begin{claim}
        \label{no3dmNoNE}
        If a 3D-matching of size $n$ does not exist in $T$, then $G_T$ has no $NE$.
    \end{claim}
    \begin{proof}
    First, it is easy to see that in every $NE$, if exists, job $f$ is first on $M_3$, the $D$-jobs are first on $|T|-n$ triplet machines, and there is exactly one $U$-job on every machine. 
        Since a matching does not exist, some element jobs from $X \cup Y \cup Z$ are not assigned on their triplet-machines. Since the $U$-jobs have higher priority than the element jobs on $M_1$ and $M_2$, at least one of these element jobs prefer $M_3$, where it is processed before a $U$-job. Now, no job from $\{u,w,v\}$ prefers $M_3$, and the three of them will be assigned on $M_1$ and $M_2$, in which they have higher priority than the $U$-job. As a result,  $G_{noNE3}$ is triggered.
    \end{proof}

    The proof of Theorem \ref{th:neISNP} follows from Claims \ref{3dmToNE} and \ref{no3dmNoNE}.    \hfill \qed

    \begin{figure}[H]
\centering
\includegraphics[width=0.8\textwidth]{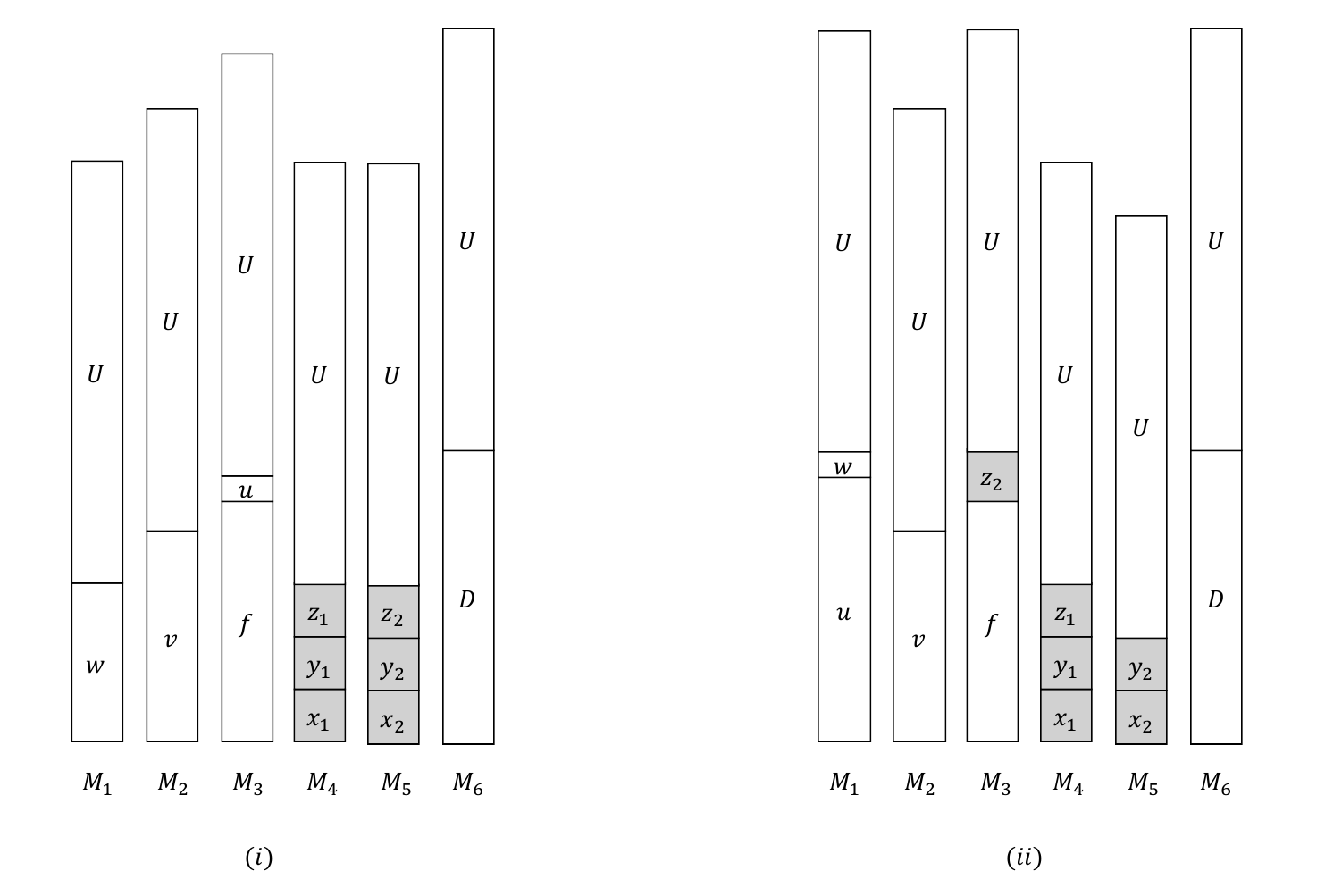}
\caption{(i) Let $T=\{(x_1,y_1,z_1),(x_2,y_2,z_2),(x_1,y_2,z_2)\}$. A matching of size $2$ exists, and job $u$ is assigned on $M_3$, resulting in a $NE$. (b) Let  $T=\{(x_1,y_1,z_1),(x_2,y_2,z_1),(x_1,y_2,z_2)\}$. A matching of size 2 does not exist, job $u$ is assigned on one of the first two machines, and $G_{noNE3}$ is induced. } 
\label{fig:NEReduction}
\end{figure}

\noindent{\bf Proof of Theorem~\ref{th:poaSym}:}
We extend the proof in~\cite{RST21} for jobs with fixed processing times. Let $\sigma$ be a schedule of symmetric jobs for which the processing time function is uniform for all jobs. The quality of $\sigma$ is characterized by the vector $(n_1(\sigma), n_2(\sigma),\ldots,n_m(\sigma))$ specifying the number of jobs on each machine.
For the class $\CidenticalJobs^+$, where all jobs have the same non-decreasing processing time function, $p(t)=b+a\cdot t$, it can be shown that the completion time of last job on a machine $j$ with $n_j$ jobs is $\frac{b((1+a)^{n_j}-1)}{a}/s_j$, and  thus the makespan of $\sigma$ is given by $\max_j \frac{b((1+a)^{n_j(\sigma)}-1)}{a}/s_j$.
For the class $\CidenticalJobs^-$, where the processing time function is $p(t)=max(\tau,b-a\cdot t)$, the completion time of the last job on a machine is $\frac{b(1-(1-a)^{m_j})}{a}/s_j + (n_j-m_j)\tau/s_j$, where the first $m_j$ have a processing time which is greater than the threshold. The makespan is $\max_j \frac{b(1-(1-a)^{m_j(\sigma)})}{a}/s_j + (n_j-m_j)\tau/s_j $.
In both cases, the completion time of the last job on machine $j$ is a non-decreasing function of $n_j$, the number of assigned jobs on machine $j$. Hence, the completion time of a machine increases with the number of assigned jobs. We denote by $c_j(n_j)$ the completion time function of machine $j$, and use it for both $\CidenticalJobs^+$ and $\CidenticalJobs^-$. The makespan is $\max_j c_j(n_j(\sigma)).$

Theorem \ref{thm:unitjobs} shows that assigning the jobs greedily, where on each step a job is added on a machine on which the cost of the next job is minimized, yields a $NE$. Let $\sigma^{\star}$ denote the resulting schedule, and let $C_1(\sigma^{\star}) \le C_2(\sigma^{\star}) \le \ldots \le C_n(\sigma^{\star})$ be the sorted vector of jobs' completion times in $\sigma^{\star}$. We show that this vector corresponds to schedules that minimize the makespan, and that  every $NE$ schedule induces the same cost vector as $\sigma^{\star}$.

Assume that there exists a schedule $\sigma'$ such that $$\max_j c_j(n_j(\sigma')) < \max_j c_j(n_j(\sigma^{\star})).$$ Let $M_1=\mbox{argmax}_j c_j(n_j(\sigma^{\star})).$ It must be that $n_{M_1}(\sigma') < n_{M_1}(\sigma^{\star})$. Since $\sum_j n_j(\sigma')=\sum_j n_j(\sigma^{\star})=n$, there must be a machine $M_2$ such that $n_{M_2}(\sigma^{\star}) < n_{M_2}(\sigma')$. 
Thus, the last job on machine $M_1$ in $\sigma^{\star}$ can benefit from migrating to machine $M_2$, as its cost will be at most $$c_2(n_{M_2}(\sigma^{\star})+1)) \le c_2(n_{M_2}(\sigma')) \le \max_j c_j(n_j(\sigma')) < \max_j c_j(n_j(\sigma^{\star})).$$ This contradicts the assumption that $\sigma^{\star}$ is a $NE$.

Now, let $\sigma$ be a $NE$ schedule with sorted cost vector and let $C_1(\sigma) \le C_2(\sigma) \le \ldots \le C_n(\sigma)$, and assume by contradiction that it has a different cost vector than $\sigma^*$. Let $i$ be the minimal index such that $C_i(\sigma^{\star}) \neq C_i(\sigma)$. Since $\sigma$ and $\sigma^{\star}$ agree on the costs of the first $i-1$ jobs, and since $\sigma^{\star}$ assigns the $i$-th job on a minimal-cost machine, it holds that $C_i(\sigma^{\star}) < C_i(\sigma)$. We get a contradiction to the stability of $\sigma$ - since some job can reduce its cost to $C_i(\sigma^{\star})$. 
\hfill \qed

\posA

\noindent{\bf Proof of Theorem~\ref{thm:PoAnotr}:}
Given $m$ and $r>1.5$, we describe a game $G_{m,r}$ with $m+1$ jobs and $m$ identical machines, such that $PoA(G_{m,r})= r$. Let $a=2r-3$. Let $N=\{j_1,\ldots, j_{m+1}\}$, and $\pi=(j_1,\ldots,j_{m+1})$. For $1 \le i \le m$, let $b_{j_i}=1$, and $a_{j_i}=0$. For job $j_{m+1}$, let $b_{j_{m+1}}=2$, $a_{j_{m+1}}=a$. In any $NE$ profile of $G_{m,r}$, each of the first $m$ jobs is first on some machine, and job $j_{m+1}$ is second on one of the machines. The makespan of any $NE$ profile is therefore $1+2+a = 3+a =2r$. 
On the other hand, in the social optimum profile, some machine processes two jobs among the first $m$ jobs, and job $j_{m+1}$ is assigned alone on one machine. The makespan of such a schedule is $2$. The social optimum profile is not a $NE$, as the second job in the pair can benefit from deviating to the  machine that processes $j_{m+1}$.
We conclude that $PoA(G_{m,r}) \ge \frac{2r}{2} = r.$
\hfill \qed

\posA

\noindent{\bf Proof of Theorem~\ref{thm:UniformRateGlobalCombined}:}
We start by proving a lower bound. 
Given $m>1$, $a$ and $n\ge2m+1$ we describe a game $G\in \CidenticalGlobalUniformRate$ played on $m$ identical machines, $M=(M_0,..,M_{m-1})$, and $n$ jobs with $a_i=a$, such that $PoA(G)=\Omega((1+a)^{\frac{n}{m}})$. We assume  that $n=m\cdot k+m+1$ for some $k\in\N$. This assumption is w.l.o.g as arbitrary small jobs can be added as first in the priority list. Let $w=(1+a)^{mk}$.

 An example of a game with $n=7,m=2$ and $a=1$ is given in Figure~\ref{fig:unboundedPoAFig}.

Consider a game with $N=X\cup Y$, where $X=(x_0,\ldots,x_{mk-1})$ and $Y=(y_0,y_1,\ldots, y_m)$. The global priority list is $\pi=(y_0,y_1,\ldots,y_{m-1},x_0,\ldots,x_{mk-1},y_m)$. The processing time function of each job in $X$ is $p_{x_i}(t)=1+at$, and the processing time function of each job in $Y$ is $p_{y_i}(t)=\frac{w-1}{a}+at=\frac{(1+a)^{mk}-1}{a}+at$. Note that $|X|=mk$, and $|Y|=m+1$.

The following simple formula is used in our analysis. For any $a>0$: \begin{equation}
\label{sumExpA}
    \sum_{\ell=0}^{n-1}{}(1+a)^\ell=\frac{(1+a)^n-1}{a}
\end{equation}

Let $\sigma^*$ be the schedule in which $y_0$ and $y_1$ are assigned on $M_0$, all the jobs in $X$ as well as $y_m$ are assigned on $M_1$, and for all $j>1$, a single $Y$-job is assigned on $M_j$ .
We show that $\sigma^*$ is optimal with respect to the minimum makespan objective. 
Consider first machines $(M_2,\ldots, M_{m-1})$, on which a single $Y$-job is assigned. The processing time of a $Y$-job when it starts at time $0$ is $\frac{w-1}{a}$, and hence the machines complete at this time.

Consider Now machine $M_0$. The first $Y$-job, $y_0$,  finishes at time $\frac{w-1}{a}$, and the processing time of the second $Y$-job, $y_1$, is $\frac{w-1}{a}+a\cdot \frac{w-1}{a}=w-1+\frac{w-1}{a}$. Therefore, $M_0$ completes at time $w-1+2\cdot \frac{w-1}{a}$.

Finally, consider $M_1$. The processing time of $x_0$ is $1$, and for all $i>1$, the processing time of $x_i$ is $p_{x_i}=1+a\cdot\sum_{\ell=0}^{i-1}p_{x_\ell}$. 
We get that $p_{x_i}=1+a\cdot\sum_{\ell=0}^{i-1}(1+a)^\ell=(1+a)^i$, and by Equation (\ref{sumExpA}), we get that the completion time of $x_i$ is $\frac{(1+a)^{i+1}-1}{a}$. In particular, the completion time of $x_{mk-1}$, which is the last $X$-job, is $\frac{(1+a)^{mk}-1}{a}=\frac{w-1}{a}$. Since $p_{y_{m}}(\frac{w-1}{a})=\frac{w-1}{a}+a\cdot \frac{w-1}{a}$, we conclude that $M_{1}$ completes, as $M_0$ completes, at time $w-1+2\cdot \frac{w-1}{a}$, and determines the makespan.

$\sigma^*$ is optimal, since in any assignment, some machine processes at least two $Y$-jobs, whose total length when processed one after the other is at least $w-1+2\cdot \frac{w-1}{a}$.

We turn to describe a possible $NE$ profile of $G$. Let $\sigma$ be a profile obtained by Algorithm \ref{alg:LS}, where ties are broken in favor of the machine with a lower index. Specifically, in $\sigma$, $\forall~i<m$, $y_i$ chooses $M_i$, then the $X$-jobs alternate between the $m$ machines, such that there are exactly $\frac{|X|}{m}=k$ $X$-jobs on each machine,  and job $y_m$ is last on $M_0$, and determines the makespan of $\sigma$. As an outcome of Algorithm \ref{alg:LS}, $\sigma$ is a $NE$.

The makespan of $\sigma$ is achieved on $M_0$. The first job on $M_0$ is $y_0$, whose processing time is $\frac{w-1}{a}$. The next $k$ jobs are the $X$-jobs $(x_0,x_m,\ldots,x_{m(k-1)})$. Denote by $C_{mi}$, the completion time of job $x_{mi}$. Job $x_0$ begins at time $\frac{w-1}{a}$, thus its processing time is $1+a\cdot (\frac{w-1}{a})=w$, and $C_0=\frac{w-1}{a}+w=\frac{w(a+1)-1}{a}$. In general, for all $1\le i \le k$, job $x_{mi}$ begins at time $C_{m(i-1)}$ and its processing time is $a\cdot C_{m(i-1)}+1$, thus $C_{mi}=a\cdot C_{m(i-1)}+1$. It is easy to verify that job $x_{mi}$ begins at time $\frac{(1+a)^{i}\cdot w-1}{a}$, has processing time of $(1+a)^i\cdot w$ and completion time $C_{mi}=\frac{w\cdot (1+a)^{i+1}-1}{a}$. In particular, $x_{m(k-1)}$, the last $X$-job on $M_0$, completes at time $C_{m(k-1)}=\frac{w\cdot (1+a)^{k}-1}{a}$. Job $y_m$ follows job $x_{m(k-1)}$ on $M_0$. It begins at time $C_{m(k-1)}$ and completes at time $a\cdot C_{m(k-1)} +\frac{w-1}{a}$, therefore, $$C_{max}(\sigma)=C_{y_m}=a\cdot \frac{w\cdot (1+a)^k-1}{a}+\frac{w-1}{a}=\frac{w-1}{a} + w\cdot (1+a)^k-1=\frac{w-1}{a}+w\cdot \sqrt[m]{w}-1.$$

We get that $PoA(G)$ is at least 
\begin{align*}
\frac{C_{max}(\sigma)}{C_{max}(\sigma^*)} = &\frac{\frac{w-1}{a}+w\cdot \sqrt[m]{w}-1}{w-1+2\cdot \frac{w-1}{a}}=\Omega(\frac{w\cdot \sqrt[m]{w}}{w-1})=\\
&=\Omega( \sqrt[m]{w})=\Omega(\sqrt[m]{(1+a)^{mk}}=\Omega((1+a)^k)=\Omega((1+a)^{\frac{n}{m}})
\end{align*}

Next, we show that this bound is tight, that is, 
    $PoA(\CidenticalGlobalUniformRate)=O((1+a)^{\frac{n}{m}})$.
The following observation is used in our proof.
\begin{observation}
    \label{cl:bound1}
    Let $B=(b_1,b_2,...,b_k)$ be the vector of {\em basic processing times} of $k$ jobs assigned on a machine $M_j$. Then, independent of the order in which the jobs are processed, the completion time of $M_j$ is at most $(1+a)^{k-1}\cdot\sum_{i}{b_i}$.
    \end{observation}
    \begin{proof}
    Assume that the jobs are processed according to their index.  Denote by $C_i$ the completion time of job $i$. The processing time of job $i$ is $a\cdot C_{i-1}+b_i$, and its completion time is $C_{i-1}+a\cdot C_{i-1} + b_i=(1+a)\cdot C_{i-1}+b_i$.  Hence, the completion time of job $k$ is $\sum_{i=0}^{k-1}{(1+a)^i\cdot b_{k-i}}$.

    Independent of the jobs' order, that is, independent of the priority list, $\pi$, for all $1 \le i \le k$, the coefficient of $b_i$ in this term is at most $(1+a)^{k-1}$. 
    Hence, for every $\pi$, the completion time of $M_j$ is at most $(1+a)^{k-1}\cdot \sum_{i=1}^{k}{b_i}.$
\end{proof}

   Let $G$ be a game played on $m$ identical machines and a global priority list $\pi$. Let $\sigma^*$ be a schedule that achieves optimal makespan, and let $\sigma$ be a $NE$ profile of $G$.  Let $B=\sum_{i=1}^{n}{b_i}$ be the sum of basic processing times of the jobs.

    Consider first the $NE$ schedule $\sigma$. Let $k_j$ denote the number of jobs assigned on machine $M_j$, and let $B_j$ denote the sum of basic processing times of jobs assigned on $M_j$, that is, $B_j=\sum_{i|\sigma_i=j}b_i$.
    Let $u$ be the job determining the makespan, and assume, w.l.o.g., that it is assigned on $M_1$. By Observation~\ref{cl:bound1}, the completion time of $M_1$ is $C_{max}(\sigma) \le (1+a)^{k_1-1} \cdot B_1$. 
    By Observation~\ref{cl:bound1}, if $u$ deviates to machines $M_j$, its completion time would be at most $(1+a)^{k_j-1+1}\cdot (B_j +b_u)=(1+a)^{k_j}\cdot (B_j +b_u)$.  Since $\sigma$ is a $NE$, such a deviation is not beneficial. Combining this with the fact that the makespan is achieved on $M_1$, we conclude that the completion time of $M_j$ in $\sigma$ is at most this term. Hence, for all $j$, $C_{max}(\sigma)\le (1+a)^{k_j}\cdot (B_j +b_u)$.

    Let $r=\argmin_j{(1+a)^{k_j}\cdot (B_j+b_u)}$. We prove the following lemma:
    \begin{lemma}
        \label{AMGM-machines}
        $(1+a)^{k_r}\cdot (B_r+b_u)\le (1+a)^{\frac{n}{m}}\cdot (\frac{B}{m}+b_u)$.

    \end{lemma}
    \begin{proof}
    Let $F_j=(1+a)^{k_j}\cdot (B_j+b_u)$. Thus,  $F_r=\min_j{F_j}$.
    Since the minimum of a set is at most the geometric mean, we have:
    $$F_r \le \left(\prod_{j=1}^{m} F_j\right)^{\frac{1}{m}} = \left(\prod_{j=1}^{m} (1+a)^{k_j} \cdot (B_j+b_u)\right)^{\frac{1}{m}}$$
    The right term can also be written as:
    $$(1+a)^{\frac{1}{m}\sum_j k_j} \cdot \left(\prod_{j=1}^{m}(B_j+b_u)\right)^{\frac{1}{m}}$$
    By the AM-GM inequality (the arithmetic mean  of a list is greater or equal to the geometric mean) we get that 
    
    $$\left(\prod_{j=1}^{m}(B_j+b_u)\right)^{\frac{1}{m}} \le \frac{1}{m}\sum_j(B_j+b_u) = \frac{B + m \cdot b_u}{m} = \frac{B}{m}+b_u$$
    Combining these, and since $\sum_jk_j=n$ we conclude  that
    $F_r
    \le (1+a)^{\frac{n}{m}}\cdot(\frac{B}{m}+b_u)$
    \end{proof}

    By Lemma~\ref{AMGM-machines} and the stability of $\sigma$, we get that 
    $$C_{max}(\sigma)\le (1+a)^{\frac{n}{m}}\cdot(\frac{B}{m}+b_u).$$
    Consider now an optimal profile $\sigma^*$. Clearly, in any schedule, and in particular in $\sigma^*$, one machine has load at least $B/m$. In addition, at least one machine has to process job $u$. Hence, $C_{max}(\sigma^*)\ge \max\{\frac{B}{m},b_u\}$, implying that $\frac{\frac{B}{m}+b_u}{C_{max}(\sigma^*)}\le2$.

    Hence, the $PoA$ is at most
    $$\frac{C_{max}(\sigma)}{C_{max}(\sigma^*)}\le\frac{(1+a)^{\frac{n}{m}}\cdot(\frac{B}{m}+b_u)}{C_{max}(\sigma^*)}\le 2\cdot (1+a)^{\frac{n}{m}}=O((1+a)^\frac{n}{m}).$$
  \hfill \qed

\posA

\posA

\noindent{\bf Proof of Theorem~\ref{poa:3-2m}:}
Let $G \in \CidenticalGlobal^{-DA}$ be a game that attains the maximal $PoA$ among all games in the class $\CidenticalGlobal^{-DA}$. Let $\sigma^*$ be an optimal (not necessarily stable) schedule of $G$. Denote by $OPT$ the makespan of $\sigma^*$. Let $\sigma$ be any $NE$ profile of $G$.
We first show that, without loss of generality, we can assume several  properties of the jobs in the game $G$. Recall that every $NE$ in this class is a result of LS run (see Theorem \ref{thm:ls_on_global}).

The following simple observation will be used in our proof. It shows that the makespan produced by LS algorithm cannot decrease if the length of a single job is extended. Formally,  

\begin{observation}
    \label{obs:start_vecs}
    Let $G\in \CidenticalGlobal^{-DA}$ be a game with $\pi=(1,2,\ldots,n)$ and processing time functions $(p_1(t),\ldots,p_n(t))$. Assume that in a run of LS (Algorithm~\ref{alg:LS}), the starting time of job $u$ is $t_u$. Consider a game $G'$ with the same priority  list and processing time functions as in $G$, except for job $u$ whose processing time function, $p'_u(t)$, is defined such that $p'_u(t_u)\ge p_u(t_u)$. Then the makespan of a schedule produced by LS on $G'$ is not smaller than the makespan produced by LS on $G$.    
    \end{observation}
\begin{proof}
    Let $C_{max}$ be the makespan of a $NE$, $\sigma$, produced by running LS on $G$. Let $u$ be a job as described above.
        Consider the game $G'$ in which the processing time function of job $u$ is replaced by $p'_u(t)$. Consider the schedule $\sigma'$ produced by running LS on $G'$ and let $C'_{max}$ be the resulting makespan.
        
        For every $1 \le i \le n$, let $L_i=(L_{i,1},..,L_{i,m})$ and $L'_i=(L'_{i,1},..,L'_{i,m})$  denote the non-increasing sorted vectors of machine loads after the assignment of job $i$ during an LS run on $G$ and $G'$ respectively.
        For every $i<u$, clearly $L_i=L'_i$. In addition,  $t_u=L_{{u-1},m}$, as this the machine is chosen by $u$. We now show by induction that for every $i\ge u$, for every $j \in M$, it holds that $L_{i,j}\le L'_{i,j}$.
        The base case is $i=u$. In a run of LS, job $u$ joins a machine with load $L_{{u-1},m}=t_u$. By our assumption on $p'_u$, we have that $p'_u(t)\ge p_u(t)$. Therefore, for all $j$, $L_{u,j}\le L'_{u,j}$. For the induction step, for every $i>u$, job $i$ joins a machine with load $L'_{{i-1},m}\ge L_{{i-1},m}$, and since $a_i\le1$, the completion time of $i$ in $\sigma'$ is not smaller than its completion time in $\sigma$. Therefore, for all $j$, $L_{i,j}\le L'_{i,j}$. 
        In particular, for $i=n$ we get that  $C'_{max} = L_{n,1} \ge L'_{n,1} = C_{max}$.
\end{proof}

\begin{observation}
    \label{obs:worst_poa_no_shorter}
        In the game $G$, w.l.o.g., the set of jobs $N$ consists of two disjoint sets $N= D \cup F$, where $D$ includes jobs that are shorter in $\sigma^*$, that is, $i \in D$ iff $t_i^\sigma < t_i^{\sigma^*}$, and $F$ includes jobs with fixed length, that is, for all $i \in F, a_i=0$.
\end{observation}
    \begin{proof}
       We first show that no job is longer in $\sigma^*$. Assume that for some job $i$ it holds that $p_i^\sigma<p_i^{\sigma^*}$. Let $p_i(t)=b_i-a_i\cdot t$, be the processing time function of $i$, and let its length in $\sigma^*$ be $y$. Consider an instance $G'$ in which job $i$ is replaced by a job with $p_i=y$ ($a_i=0$). The optimal schedule $\sigma^*$ remains valid, therefore, $OPT(G') \le OPT(G)$. In a run of LS on $G'$, job $i$ is longer at the time it is assigned, therefore, by Observation~\ref{obs:start_vecs}, the makespan of $G'$ cannot decrease, and $PoA(G')\ge PoA(G)$.

       Given that no job is longer in $\sigma^*$, any job is either shorter in $\sigma^*$ and therefore belongs to $D$, or has the same length in both schedules, and w.l.o.g., has a fixed length and belongs to $F$. Note that if $a_i>0$ and $t_i(\sigma)=t_i(\sigma^*)=x$ (possibly $x=\tau_i$) then job $i$ can be replaced by a job $i'$ whose length function is $p_i(t)=x$ (and $a_i=0$). Again, this modification does not hurt $OPT$ and does not affect the outcome of LS.
    \end{proof}
    \begin{observation}
        \label{obs:Rate1ForGlobal}
        In the game $G$, w.l.o.g., for every job $u_i\in D$, it holds that $a_i=1$.
    \end{observation}

        \begin{proof}
         Let $u_i$ be a job with $p_i(t)=b_i-a_i\cdot t$, for some  $a_i<1$. Assume that the length of $u_i$ in $\sigma^*$ is $x$ and its length in $\sigma$ is $x+\beta$. Since $a_i<1$, it must be that $t_i^{\sigma^*}-t_i^{\sigma} >\beta$. Consider an instance in which job $u_i$ is replaced by a job $i'$ with length function $p_{i'}(t)=(x +t_i^{\sigma^*})-t$ (note that $a_{i'}=1$). The makespan of the optimal schedule remains the same, because job $i'$ has length $x$. In $\sigma$, job $i'$ has the same starting time as $i$, but its length is $x+t_i^{\sigma^*}-t_i^{\sigma}>x+\beta$, therefore, by Observation~\ref{obs:start_vecs}, the makespan cannot decrease, and the $PoA$ is not smaller.
        \end{proof}

Based on the above properties of $G$ we can bound the total processing time of jobs in $\sigma$.
\begin{lemma}
    \label{ratioOPTMaxGlobal}
$\sum_i p_i^{\sigma} \le (2m-1)\cdot OPT.$  
\end{lemma}
\begin{proof}
We begin by bounding the maximum difference, $\Delta_G$ between the total processing time of the jobs in $\sigma$ and in $\sigma^*$, that is, $\Delta_G= \sum_i p_i^{\sigma}-\sum_i p_i^{\sigma^*}$.
Recall that $D$ is the set of jobs for which $a_i=1$ that have earlier starting time, and thus longer length, in $\sigma$. All other jobs have fixed lengths. 

Let $D=\{u_1,\ldots,u_d\}$, and assume that $\pi(u_i)<\pi(u_{i'})$ for $i<i'$. Consider a job $u_i \in D$. Denote by $t_i^{\sigma^*}$ and  $t_i^{\sigma}$  the start time of $u_i$ in $\sigma^*$ and in $\sigma$, respectively. 
  By Observation~\ref{obs:Rate1ForGlobal}, $a_i=1$.
Note that for any job with $a_i=1$, the completion time of job $i$ is exactly $b_i$ if it starts being processed by time $b_i$, and is more than $b_i$ if it starts being processed later than $b_i$ (and has length $\tau_i$). Therefore, for every $u_i \in D$, it holds that $OPT \ge b_i$.
Let $\delta_i = p_i^\sigma-p_i^{\sigma^*}$. Clearly, $\delta_i \le b_i$, and by the above, $\delta_i \le OPT$.  
Hence, if $|D|\le m-1$, then $\Delta_G$ is at most $(m-1)\cdot OPT$. 
\hfill\qed

    \begin{claim}
        \label{obs:globalD=m}
    In the $NE$ schedule $\sigma$, there is at most one job from $D$ on each machine.
    \end{claim}
    \begin{proof}
     Recall that a job with $p_i=max\{b_i-t,\tau_i\}$ has a completion time $b_i$ regardless of its starting time, unless it reaches length $\tau_i$. Given that jobs in $D$ are longer in $\sigma$, none of them reaches length $\tau_i$, and thus, every job $u_i \in D$ completes at time $b_i$ in $\sigma$.
     
    Assume by contradiction that there is a machine, $j$, that processes at least two jobs from $D$ in $\sigma$. Let $u_k$ be the last job from $D$ on machine $j$. By the above, its completion time is $b_k$. Hence, we can remove from the instance all jobs assigned in $\sigma$ on machine $j$ before $u_k$. The makespan of $\sigma$ remains the same. The makespan of an optimal schedule can only reduce, thus, the $PoA$ can only increase.
    \end{proof}

    Claim~\ref{obs:globalD=m} implies that if $|D| \ge m$, then w.l.o.g., $|D|=m$. In addition, let $m_d$ denote the number of machines in $\sigma$ to which jobs from $D$ are assigned. Clearly, if $m_d\le m-1$ then $\Delta_G$ is at most $(m-1)\cdot OPT$. Hence, we can conclude that  there is exactly one job from $D$ on each machine. Consider the $NE$ profile $\sigma$. Denote by $B$ the set of jobs with fixed processing time assigned before some job from $D$ in $\sigma$. Let $P_B=\Sigma_{i\in B} b_i$.
    Recall that $\Delta_G$ is affected only by jobs in $D$, and that each job in $D$ is shortened by exactly its starting time in $\sigma$, which is determined by jobs in $B$ on its machine. Therefore, the total length of $D$-jobs in $\sigma$ is $(\sum_{i \in D} b_i) - P_B$.

    Next, we provide a lower bound on the total length of the $D$-jobs in $\sigma^*$. We first show the following property of the priority list $\pi$.
    \begin{claim}
        \label{cl:no_green_before_d1}
        Let $u_1$ be the first job from $D$ in $\pi$. All jobs that precede $u_1$ in $\pi$ belong to $B$.
    \end{claim}
    \begin{proof}
        Assume by contradiction that $\pi(v)<\pi(u_1)$ for some job $v\notin B$. By the choice of $u_1$, we have that $v\notin D$. Consider an execution of the LS algorithm. Since $v$ appears before $u_1$ in $\pi$, no job from $D$ is assigned before $v$ and $u_1$ are added. Let $M_j$ denote the machine to which $v$ is assigned. In the final $NE$, $M_j$ ends up with one job from $D$, which is assigned after $v$. This contradicts the definition of $B$, according to which any job assigned before a job from $D$ on its machine must belong to $B$. Hence, all jobs preceding $u_1$ in $\pi$ must be from $B$.
        \end{proof}

    Denote by $B_1\subseteq B$ the set of jobs in $B$ with higher priority than $u_1$. Consider the optimal schedule $\sigma^*$. By Claim~\ref{cl:no_green_before_d1}, we know that $u_1$ is delayed by at most $\sum_{i\in B_1}p_i \le P_B$, and hence its length is at least $b_1-P_B$. For $2 \le i \le d$, the length of job $u_i$ in $\sigma^*$ is clearly non-negative.
    Therefore, $\Delta_G= \sum_i p_i^{\sigma}-\sum_i p_i^{\sigma^*} \le (\sum_{i \in D} b_i) - P_B - (b_1-P_B) = b_2+\ldots+b_m \le (m-1)\cdot OPT$.

Given the bound on $\Delta_G$, the total job lengths in $\sigma$ can be bounded as follows. 
$$\sum_i p_i^{\sigma} \le \sum_i p_i^{\sigma^*} + \Delta_G \le \sum_i p_i^{\sigma^*}+ (m-1)\cdot OPT \le m\cdot OPT + (m-1)\cdot OPT=(2m-1)\cdot OPT.$$
\end{proof}

We can now conclude the $PoA$ analysis:
Let $C_{max}(\sigma)$ denote the makespan of $\sigma$. Let $L_j$ denote the total load on the $j^{th}$ machine, and let $k$ be the job that finishes last and determines $C_{max}(\sigma)$. Since $\sigma$ is a possible outcome of LS algorithm, all the machines were busy when $k$ starts its processing. Thus $\forall j, L_j\ge C_{max}(\sigma)-p^{\sigma}_k$. For at least one machine, $L_j=C_{max}(\sigma)$. Summing up, we get:
$$    \sum_i p^{\sigma}_i=\sum_j L_j\ge (m-1)(C_{max}(\sigma)-p^{\sigma}_k)+C_{max}(\sigma).$$
This implies that 
$$\sum_i p^{\sigma}_i +(m-1)\cdot p^{\sigma}_k\ge m\cdot C_{max}(\sigma).$$ Therefore, $$
   C_{max}(\sigma)\le \frac{1}{m} \sum_i p^{\sigma}_i + p^{\sigma}_k \cdot \frac{m-1}{m}. $$

Now, consider the optimal schedule $\sigma^*$. The start time of job $k$, which determines $C_{max}(\sigma)$, is clearly earlier in $\sigma^*$. Therefore, $p_k^{\sigma} \le p_k^{\sigma^*}$, implying that 
$OPT \ge p^{\sigma}_k$.
In addition, by Lemma~\ref{ratioOPTMaxGlobal}, we have that
$\sum_i p^{\sigma}_i \le (2m-1)\cdot OPT$. We conclude that 
$$C_{max}(\sigma)\le \frac{m-1}{m}\cdot OPT+ \frac{2m-1}{m}\cdot OPT = (3-\frac{2}{m})\cdot OPT.$$

Next, we show that the analysis is tight. Specifically,
For any $0<\epsilon <1$ and $m \ge 2$, we present a game in $\CidenticalGlobal^{-DA}$ played on $m$ identical machines, such that $PoA(G) = 3-\frac{2}{m}-\epsilon$.    An example for $m=3$ is given in Figure~\ref{fig:global_3m-1}.

Given $m$ and $\epsilon$, let $\epsilon'$ be a small constant such that $\frac{3-\frac{2}{m}-2\epsilon'}{1+(m-1)\epsilon'}=3-\frac 2 m -\epsilon$. We define the following game, $G$, played on $m$ identical machines.

The set $N$ consists of three sets of jobs. The first set includes a single job, $a$, for which $p_a(t)=1-\epsilon'$. The second set is $U=\{u_1,\ldots,u_{(m-1)m}\}$. For every $u_i \in U$, let $p_{u_i}(t)=\frac{1-\epsilon'}m$. The third set is $V=\{v_1,\ldots,v_m\}$. For every $v_i \in V$, let $p_{v_i}(t) = \max\{1-t,\epsilon'\}$. Note that Job $a$ as well as the $U$-jobs have fixed processing times, while the $V$-jobs have decreasing processing times.

Let $U_1=\{u_1,u_2,\ldots,u_{m}\}$. The global priority list is $\pi=(U_1,V,U\setminus U_1,a)$. Whenever the list includes a set, the set elements are ordered according to their index.
    
In a possible optimal schedule, the first machine processes Job $a$. For every $2 \le j \le m$, machine $j$ processes the $m$ jobs $\{u_{m(j-2)+1},\ldots,u_{m(j-1)}\}$. In addition, machine $M_2$ processes all $m$ V-jobs.
Note that the $m$ $U$-jobs on $M_2$ are the jobs of $U_1$, that precede the $V$-jobs in $\pi$, therefore, each of the $V$-jobs starts its processing at time at least $1-\epsilon'$, and has length $\epsilon'$. The makespan of the optimal schedule is $\frac{m\cdot (1-\epsilon')}{m}+m\cdot\epsilon'=1+(m-1)\cdot\epsilon'=1+\frac{\epsilon}{3-\epsilon}$  
This schedule is optimal, since the fixed-length jobs are balanced, and each of the $V$-jobs has its minimal possible length. It is easy to see that any schedule in which some $V$-jobs are on the same machine with a job in $U \setminus U_1$ has a higher makespan.

On the other hand, the following schedule, $\sigma$, is a $NE$: 
for every $1 \le j \le m$, machine $j$ processes job $u_j$, followed by job $v_j$, as well as the set $\{u_{mk +j}\}_{1 \le k \le m-2}$ which includes $m-2$ jobs from $U \setminus U_1$. In addition, the first machine processes also Job $a$. Each of the $V$-jobs begins its processing at time $\frac{1-\epsilon'}{m}$ and therefore has length $1-\frac{1}{m}+\frac{\epsilon'}{m}$.  
The makespan of this schedule is obtained on the first machine and is equal to $(m-1)\cdot\frac{1-\epsilon'}{m}+1-\frac{1}{m}+\frac{\epsilon'}{m}+1-\epsilon'=3-\frac{2}{m}-2\epsilon'(1-\frac 1 m)$.

We show that $\sigma$ is a $NE$. Clearly, the jobs of $U_1$ are all first, and therefore stable. Every $V$-job is assigned after a single job of $U_1$, and since jobs of $U_1$ have the highest priority on all the machines, no $V$-job can improve its assignment. The remaining $m(m-2)$ $U$ jobs have a lower priority than the $V$-jobs, and are assigned in a balanced way ordered by their indices; therefore, no $U$-job has a beneficial deviation. Finally, job $a$ has the lowest priority, and will have the same completion time in all machines.

We conclude that $PoA(G)=\frac{3-\frac{2}{m}-2\epsilon'}{1+(m-1)\epsilon'}=3-\frac 2 m -\epsilon$.
\hfill \qed

\posA
\noindent{\bf Proof of Claim~\ref{cl:SDRdensity}:}
We define $w_k = W_k(\sigma_{SDR})$ as the total cumulative machine work performed up to the moment the $k$-th job begins its processing in the SDR schedule. Note that for $1 \le k \le m$, $w_k = 0$. We prove the claim by induction on $k$. Specifically, we show that for every $k=1,\ldots,n$, there exists a profile that maximizes the integral $\int_0^{w} \mathcal{A}_{\sigma}(u) du$ for all $w \le w_k$ and processes the first $k$ jobs in SDR order.

{\bf Base Case:} $w \le w_{m+1}$. In SDR, the first $m$ jobs assigned when $w=0$ are those with the minimal $a_i$. Thus, for any $w \le w_{m+1}$, the set of jobs in $B_{SDR}(w)$ is $\{m+1, m+2, \dots, n\}$, which contains the $n-m$ jobs with the maximal $a_i$ values in the system.
For any other schedule $\sigma \neq \sigma_{SDR}$, with no-intended idle, $m$ jobs are processed at $w=0$. Therefore, $n-m$ jobs must be in $B_{\sigma}(w)$. Since SDR chooses the $m$ smallest $a_i$ values to process, we get, $\mathcal{A}_{SDR}(u) \ge \mathcal{A}_\pi(u)$ for all $u \in [0, w_{m+1}]$ .

Therefore, $\int_0^w \mathcal{A}_{SDR}(u) du \ge \int_0^w \mathcal{A}_\sigma(u) du$ for all $w \in [0, w_{m+1}]$.

{\bf Induction Step:}  Assume the claim holds for all $w \le w_k$. We show it holds from $w_k$ to $w_{k+1}$. By the induction hypothesis, there exists a profile $\sigma$ that maximizes the cumulative deterioration density and agrees with SDR on the first $k-1$ jobs. Thus, $B_{\sigma}(w_k)=B_{SDR}(w_k)$. 

At workload $w_k$, a machine finishes a job, and SDR selects the job $k$ with the smallest deterioration rate from $B_{SDR}(w_k)$.
For any $u \in [w_k, w_{k+1}]$, $B_{SDR}(u)$ contains the $n-k$ jobs with the largest possible sum of $a_i$ among all unprocessed jobs. Consider any other profile $\sigma'$ that deviates from SDR at $w_k$ by processing a job $J_\ell$ with $a_\ell > a_{k}$. For all $u$ such that $J_\ell$ is being processed in $\sigma'$ but $J_{k}$ is being processed in $SDR$, we have $\mathcal{A}_{\sigma'}(u) = \sum_{i \in B_{\sigma'}(u)} a_i = (\sum_{i \in B(w_k)} a_i) - a_\ell < (\sum_{i \in B(w_k)} a_i) - a_{k} = \mathcal{A}_{SDR}(u)$. 

Since $\mathcal{A}_{SDR}(u) \ge \mathcal{A}_{\sigma'}(u)$ for all $u \in [w_k, w_{k+1}]$, the induction step is proved, and we conclude that SDR policy maximizes  $\int_0^w \mathcal{A}_{\sigma}(u) du$ for all $0 \le w \le W_n(\sigma_{SDR})$.
\hfill\qed

\posA
\noindent{\bf Proof of Claim~\ref{cl:S_nleOPT}:}
Let $W_{total}(\sigma^*)$ be the total processing time of all jobs $\sigma^*$.  In the optimal schedule $\sigma^*$, $W_{total}(\sigma^*) \le m \cdot OPT$. 

For the optimal schedule, it holds that \begin{equation}
    \label{eq:OPTB}
    W_{total}(\sigma^*) = \sum_{i=1}^n b_i - \sum_{i=1}^n a_i S_i^{\sigma^*} =  \sum_{i=1}^n b_i - \frac{1}{m} \int_0^{W_{total}(\sigma^*)} \mathcal{A}_{\sigma^*}(u) du.
\end{equation}
    This is true since job $i$ contributes to the integral $a_i \cdot W_i = a_i \cdot (m S_i^{\sigma^*})$.

Similarly, for $\sigma_{SDR}$, 
\begin{equation}
    \label{eq:SDR_B}
    m \cdot S_n^{\sigma_{SDR}} = \sum_{i=1}^{n-1} p_i(S_i^{\sigma_{SDR}}) = \sum_{i=1}^{n-1} (b_i - a_i S_i^{\sigma_{SDR}}) = \sum_{i=1}^{n-1} b_i - \frac{1}{m} \int_0^{m S_n^{\sigma_{SDR}}} \mathcal{A}_{\sigma_{SDR}}(u) du.
\end{equation}

Assume by contradiction that $S_n^{\sigma_{SDR}} > OPT$. This implies $W_n(\sigma_{SDR})=m \cdot S_n^{\sigma_{SDR}} > m \cdot OPT \ge W_{total}(\sigma^*)$. Combining the fact that $\mathcal{A}_{\sigma_{SDR}}(u)$ is non negative with Claim~\ref{cl:SDRdensity} for $w=W_{total}(\sigma^*)$, we get, $$\frac{1}{m} \int_0^{m \cdot S_n^{\sigma_{SDR}}} \mathcal{A}_{\sigma_{SDR}}(u) du  \ge \frac{1}{m} \int_0^{W_{total}(\sigma^*)} \mathcal{A}_{\sigma_{SDR}}(u) du > \frac{1}{m} \int_0^{W_{total}(\sigma^*)} \mathcal{A}_{\sigma^*}(u) du.$$

Adding $m \cdot S_n^{\sigma_{SDR}} >  W_{total}(\sigma^*)$, we obtain,
$$m \cdot S_n^{\sigma_{SDR}} + \frac{1}{m} \int_0^{m \cdot S_n^{\sigma_{SDR}}} \mathcal{A}_{\sigma_{SDR}}(u) du > W_{total}(\sigma^*) + \frac{1}{m} \int_0^{W_{total}(\sigma^*)} \mathcal{A}_{\sigma^*}(u) du.$$

Combining with Equations~(\ref{eq:OPTB}) and (\ref{eq:SDR_B}), we conclude $\sum_{i=1}^{n-1} b_i > \sum_{i=1}^n b_i$, which is a contradiction since $b_n \ge 0$. Thus, $S_n^{\sigma_{SDR}} \le OPT$.
\hfill\qed

\posA
\noindent{\bf Modified policy for instances where jobs may reach their threshold length:} The above assumes that no job reaches its threshold processing time $\tau$. This assumption is crucial for Equation~\ref{eq:OPTB} since if a job reaches $\tau$,  its effective deterioration rate drops to zero, breaking the density function $A_\sigma(w)$.
However, it is possible to maintain a $PoA$ of $2$ even if jobs do reach $\tau$, by designing a dynamic mechanism that generalizes the $SDR$ idea. In this dynamic mechanism, we define the \emph{actual deterioration rate} of a job $k$ at time $t$ as its real  shrinkage per unit of time. Formally:

Let $t$ be the current time (a machine becomes available), and let $t'$ be the next decision point (where a current running job will complete). The \emph{actual deterioration rate} of an unassigned job $k$ at this interval is defined as $a_k^{actual}=\frac{p_k(t)-p_k(t')}{t'-t}$. Note that $a_k^{actual} = a_k$ if job $k$ does not reach its threshold length by time $t'$ and $a_k^{actual} < a_k$ if it does.

The first $m$ jobs in this policy are the jobs with the minimal rate $a$. Indeed, when starting at $t=0$, the threshold length is not relevant. Once $m$ jobs are assigned and all the machines are busy, and the time $t'$ in which a machine will become available can be calculated. The mechanism dynamically selects the unassigned job with the minimal $a_k^{actual}$ with respect to $t'$. By dynamically selecting the job with the minimal $a_k^{actual}$ at every decision point, the density of the waiting jobs is maximized with respect to the true shrinkage. Consequently, under this dynamic policy, Claim~\ref{cl:SDRdensity} holds for every workload $w$ that corresponds to a "stair" (completion of some processed job). Formally, for any $0 \le w \le W_n(\sigma_{SDR}^{actual})$, if some job completes at workload $w$, then the dynamic policy maximizes the cumulative deterioration density $\int_0^w \mathcal{A}_{\sigma}(u) du$, among all profiles with no intended idles.
\hfill \qed

\posA
\noindent{\bf Proof of Theorem~\ref{th:sdr2}(lower bound):}
    Given $m$ and  $\delta$ we define the following game, $G^{SDR}$, played on $m$ identical machines. First, we define two values, $b$ and $B$, such that $B>>b$.
    We construct the game such that in the optimal schedule $\sigma^*$, the makespan is determined by a single job with a large basic processing time, $B$. The remaining $n-1$ jobs have small basic processing times, $b$, and are divided among $m-1$ machines. 
    Hence, there are $k=\frac{n-1}{m-1}\in\N$ small jobs on each of the corresponding $m-1$ machines.
    We denote the set of small jobs by $U={u_1,u_2,\ldots,u_{n-1}}$. The processing time of each job in this set is $p_u(t)=max\{b-a\cdot t,\tau\}$. The processing time of the large job, denoted by $v$, is $p_v(t)=max\{B-a\cdot t,\tau\}$.
    The deterioration rate of all jobs is $a$, where $a\to 0$. Specifically, $a$ is small enough such that no job hits the minimal time $\tau$ in the optimal schedule.
    A priority list that obeys the proposed coordination mechanism is $\pi=(U,v)$.
    
    It is easy to see that the completion time of the $k$-th small job, denoted by $T(k)$, is $\frac{b}{a}\cdot (1-(1-a)^k)$.
    For a small value of $\epsilon >0$, we determine $k$ such that $T(k)\in[B(1-\epsilon),B]$, hence 
\begin{eqnarray}
\label{kSmallJobs}
\frac{\ln\!\left(1-\frac{aB(1-\varepsilon)}{b}\right)}{\ln(1-a)}
\;\le\;
k
\;\le\;
\frac{\ln\!\left(1-\frac{aB}{b}\right)}{\ln(1-a)}.
\end{eqnarray}

This interval is nonempty provided that $\frac{b}{a} > B$, which we assume henceforth.

Thus, each machine processing only small jobs completes by time at most $T(k)\le B$, while the machine processing $v$ completes exactly at time $B$.

Consider now a $NE$ profile $\sigma$. In a profile produced by LS Algorithm~\ref{alg:LS}, the $U$ jobs are assigned in a balanced way among $m$ machines, and the large job $v$ is assigned last on one of the machines. We assume w.l.o.g that $\frac{n-1}{m}=k'$ for some $k'\in \N$.
The completion time of the $k'$-th job is $T(k')=\frac{b}{a}\cdot (1-(1-a)^{k'})$, implying that the completion time of job $v$ is $T(k')+B-a\cdot T(k')=T(k')\cdot (1-a) + B$.

Now, recall that $T(k)\ge B(1-\epsilon)$. Additionally, we assume $\frac{b}{a}=B(1+o(1))$ as $a\to 0$,  implying that $b-aB = o(b)$.

Rearranging $T(k)$, we get $(1-a)^k\le\frac{b-a\cdot B(1-\epsilon)}{b}=\frac{b-a\cdot B +\epsilon\cdot ab}{b}$. Hence $(1-a)^k\le O(\epsilon)$.
Recall that $k'=\frac{m-1}{m}\cdot k$, so $(1-a)^{k'}=((1-a)^k)^{\frac{m-1}{m}}$. Combining with the previous equation, we get that  $(1-a)^{k'}\le O(\epsilon^{\frac{m-1}{m}})$, and $T(k') = B\left(1-O(\epsilon^{\frac{m-1}{m}})\right).
$

For any $\delta>0$, choosing $\epsilon$ sufficiently small ensures that $T(k')\ge B(1-\delta)$.

Concluding the above, we get that $C_{max}(\sigma)\ge  B(1-\delta)\cdot (1-a) +B$. Since $a\to 0$, we obtain $C_{\max}(\sigma)\ge B(2-\delta)$ for any $\delta>0$.
And,
$$PoA(G^{SDR})=\frac{C_{max}(\sigma)}{C_{max}(\sigma^*)}=\frac{B(2-\delta)}{B}\ge 2-\delta.$$

The following example depicts a game with $4$ identical machines with $PoA>1.99$. A visualization of the game is given in Figure~\ref{fig:poa1.99Neg}.

    Let $b=10$, $a=0.009$, $B=1111$ and $\tau=0.001$. Consequently, $p_v=max\{10-0.009\cdot t,0.001\}$ and $p_u=max\{1111-0.009\cdot t,0.001\}$.
    To satisfy Equation~(\ref{kSmallJobs}), we choose $k=1016$, $k'=762$, and $n=1016\cdot 3 +1 =3049$.
    The makespan of the optimal profile, where the large job is assigned alone on one of the machine is $1111$, as $T(1016)=\frac{10}{0.009}\cdot (1-(1-0.009)^{1016})\le 1111$
    .

    Consider now the $NE$ profile. The completion time of the $k'$-th small job is $T(762)=\frac{10}{0.009}\cdot (1-(1-0.009)^{762})\le 1110$, and the processing time of job $u$ is $1111-0.009\cdot 1110=1101.1$. 
    Therefore, the $PoA$ of the game is $\frac{1101.1+1110}{1111}\approx 1.99$.

\begin{figure}[h]
\centering
\includegraphics[width=\textwidth]{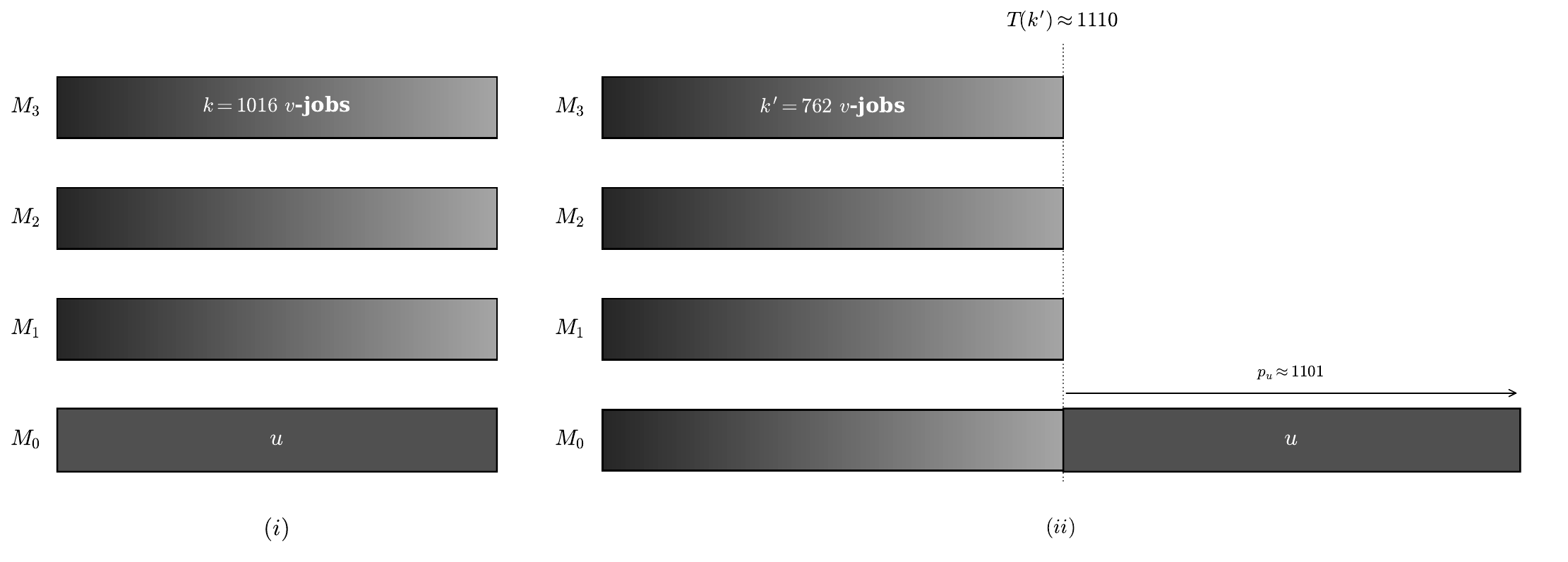}
\caption{(i) An optimal profile with makespan $\approx 1111$.(ii) A $NE$ profile with makespan $\approx 1101.1+1110$.}
\label{fig:poa1.99Neg}
\end{figure}
\hfill \qed

\posA
\noindent{\bf Proof of Theorem~\ref{th:LBDR}:}
Let $R_i=\frac{b_i}{a_i}$, and $\pi=(1,2,\ldots,n)$, such that $R_1\ge R_2 \ge\ldots,R_n$. 
Assume that an optimal profile $\sigma^*$ achieves a makespan of $OPT$.

Consider a NE profile $\sigma$. W.l.o.g., job $n$ is the job that determines the makespan of $\sigma$. Assume that it starts its processing at time $S_n$, and that it does not reach its threshold processing time $\tau$. Hence, $C_{max}(\sigma)=S_n + p_n(S_n)=S_n +b_n-a_n\cdot S_n$. Since $p_n(S_n)>0$,  $S_n<\frac{b_n}{a_n}=R_n$.

All jobs are delay-averse, meaning that $a_i\le1$. Since there is a schedule with makespan $OPT$, it must be that for every job $i$, $OPT\ge t_i^{\sigma^*} + b_i-a_i\cdot t_i^{\sigma^*}=b_i + t_i^{\sigma^*}(1-a_i)\ge b_i$.
In addition,
$$C_{max}(\sigma) = b_n + S_n(1-a_n)<b_n + R_n(1-a_n)=b_n +\frac{b_n}{a_n}(1-a_n) = b_n + \frac{b_n}{a_n} -b_n =\frac{b_n}{a_n}=R_n.$$
If $C_{max}(\sigma)\le OPT$, then $PoA(G^{LBDR})=1$ and we are done. Otherwise, we can assume below that $C_{max}(\sigma)>OPT$, which implies $R_n>OPT$.

Since $n$ is the last job in $LBDR$ order, for every $i<n$, $R_i\ge R_n$. Therefore, $a_i=\frac{b_i}{R_i}\le\frac{b_i}{R_n}$. Assume that in some profile job $i$ finishes at time $C_i$.
The processing time of such job satisfies $$p_i(t_i)=b_i-a_i\cdot t_i\ge b_i-\frac{b_i}{R_n}\cdot t_i=b_i\cdot (1-\frac{t_i}{R_n}),$$ and hence, 
\begin{equation}
    \label{eq:RminusC}
    R_n-C_i=R_n-t_i-p_i(t_i)\le R_n-t_i-b_i\cdot (1-\frac{t_i}{R_n})=(R_n-t_i)\cdot (1-\frac{b_i}{R_n}).
\end{equation}
Consider again the optimal profile $\sigma^*$, and a machine $M_j^*$ with a completion time of $L_j^*$.

Let the jobs assigned to $M_j^*$ be ordered according to their processing order on the machine, 
and denote them by $o_1,o_2,\ldots,o_k$. Let $C_{\ell}$ be the completion time of job $o_\ell$.

We prove the following claim:
\begin{claim}
\label{lem:volumeBoundOnAmachine}
For every $1 \le \ell \le k$ it holds that $1-\frac{C_{{o_{\ell}}}}{R_n}\le \prod_{h=1}^{\ell}\left(1-\frac{b_{o_h}}{R_n}\right).$
\end{claim}
\begin{proof}
The proof is by induction on $\ell$.

\textbf{Base case ($\ell=1$).}  
The first job, $o_1$, starts at time $t_{o_1}=0$. Hence
$C_{o_1}=p_{o_1}(0)=b_{o_1}.$
Therefore,
\[
1-\frac{C_{o_1}}{R_n}
=
1-\frac{b_{o_1}}{R_n}
=
\prod_{h=1}^{1}\left(1-\frac{b_{o_{h}}}{R_n}\right).
\]

\textbf{Induction step.}  
Assume the claim holds for $\ell-1$.  
The starting time of job ${o_{\ell}}$ is $t_{o_{\ell}}=C_{{{o_{\ell}}-1}}$. From Equation~\ref{eq:RminusC},
\[
R_n-C_{o_{\ell}}
\le
(R_n-t_{o_{\ell}})\left(1-\frac{b_{o_{\ell}}}{R_n}\right).
\]
Substituting $t_{o_{\ell}}=C_{{o_{\ell}}-1}$ gives
\[
R_n-C_{o_{\ell}}
\le
(R_n-C_{{o_{\ell}}-1})\left(1-\frac{b_{o_{\ell}}}{R_n}\right).
\]
Dividing by $R_n$ yields
\[
1-\frac{C_{o_{\ell}}}{R_n}
\le
\left(1-\frac{C_{{o_{\ell}}-1}}{R_n}\right)
\left(1-\frac{b_{o_{\ell}}}{R_n}\right).
\]
By the induction hypothesis,
\[
1-\frac{C_{{o_{\ell}}-1}}{R_n}
\le
\prod_{h=1}^{{\ell}-1}\left(1-\frac{b_{o_{h}}}{R_n}\right).
\]
Substituting gives
\[
1-\frac{C_{o_{\ell}}}{R_n}
\le
\prod_{h=1}^{{\ell}}\left(1-\frac{b_{o_{h}}}{R_n}\right).
\]
Applying the claim to the last job on machine $M_j^*$ yields
\[
1-\frac{L_j^*}{R_n}\le \prod_{i:\sigma^*_i=j}\left(1-\frac{b_i}{R_n}\right).
\]
\end{proof}

Since the makespan of $\sigma^*$ is $OPT$, the completion time of any machine satisfies $L_j^*\le OPT$, therefore
$$1-\frac{OPT}{R_n}\le \prod_{i:\sigma^*_i=j}(1-\frac{b_i}{R_n}).$$

Multiplying over all $m$ machines yields the product over all jobs in the instance: 
$$\left(1-\frac{OPT}{R_n}\right)^m \le \prod_{i=1}^{n}\left(1-\frac{b_i}{R_n}\right) = \left(1-\frac{b_n}{R_n}\right)\cdot \prod_{i=1}^{n-1}\left(1-\frac{b_i}{R_n}\right).$$

Now, recall that $\sigma$ is a possible outcome of LS (\ref{alg:LS}), and hence when job $n$ chooses its machine it selects the machine with the minimal load. Therefore, the load on every machine before job $n$ is assigned is at least $S_n$. 
The total processing time of jobs assigned before job $n$ in $\sigma$ is $\sum_{i=1}^{n-1}p_i^{\sigma}\ge m\cdot S_n$. The actual processing times of the jobs is smaller than their basic processing time, hence $\sum_{i=1}^{n-1}b_i\ge m\cdot S_n$.

Consider again $\prod_{i=1}^{n-1}(1-\frac{b_i}{R_n})$. It is known that for any real value $u$, $1-u\le e^{-u}$. Hence, using this rule, and sum of exponents rule, we get that
$$\prod_{i=1}^{n-1}(1-\frac{b_i}{R_n})\le \exp(-\sum_{i=1}^{n-1}\frac{b_i}{R_n})\le\exp(-\frac{mS_n}{R_n}).$$

We conclude that $$(1-\frac{OPT}{R_n})^m\le(1-\frac{b_n}{R_n})\cdot\exp(-\frac{mS_n}{R_n}).$$
By taking the natural logarithm ($b_n\le OPT$ and $R_n>OPT$ so $1-\frac{OPT}{R_n}>0$ and $1-\frac{b_n}{R_n}>0$) and rearranging we get $$S_n\le-R_n\cdot\ln(1-\frac{OPT}{R_n})+\frac{R_n}{m}\cdot \ln(1-\frac{b_n}{R_n}).$$

We now have an upper bound on the start time of the last job $n$ in the $NE$ profile $\sigma$. We next show that the makespan of $\sigma$ can be described as a function of the basic processing time of job $n$, $b_n$.
Recall that $C_{max}(\sigma)=S_n(1-\frac{b_n}{R_n})+b_n$. So, the makespan as a function of $b_n$ is given by  $$C_{max}(\sigma)(b_n)\le [-R_n\cdot\ln(1-\frac{OPT}{R_n})+\frac{R_n}{m}\cdot \ln(1-\frac{b_n}{R_n})]\cdot (1-\frac{b_n}{R_n})+b_n.$$

The second derivative with respect to $b_n$ is $\frac{1}{m(R_n-b_n)}$. Since $R_n>OPT\ge b_n$, this is strictly positive, meaning that the function is strictly convex.
Hence, we can find the maximum value of the function by calculating its values at the boundaries, $0$ and $OPT$. 

\textbf{Case 1:} $b_n=OPT$.

Substituting $b_n=OPT$ to $C_{max}(\sigma)(b_n)$ yields
$$C_{max}(\sigma)(OPT)\le[-R_n\cdot \ln(1-\frac{OPT}{R_n})+\frac{R_n}{m}\cdot \ln(1-\frac{OPT}{R_n})]\cdot (1-\frac{OPT}{R_n})+OPT,$$
implying that 
$$C_{max}(\sigma)(OPT)\le [(-R_n\cdot (1-\frac{1}{m}))\cdot \ln(1-\frac{OPT}{R_n})]\cdot (1-\frac{OPT}{R_n})+OPT.$$
Taking $OPT$ out:
$$C_{max}(\sigma)(OPT)\le OPT\cdot [-(1-\frac{1}{m})\cdot \frac{R_n-OPT}{OPT}\cdot \ln (1-\frac{OPT}{R_n})+1].$$

Note that for $u\in(0,1), -\frac{1-u}{u}\cdot \ln(1-u) <1$. Set $u=\frac{OPT}{R_n}$: $$C_{max}(\sigma)(OPT)\le OPT\cdot [(1-\frac{1}{m})\cdot 1+1]=(2-\frac{1}{m})\cdot OPT.$$

We note here that a simple instance with a $PoA$ of $2-1/m$ can be generated by assigning each of the first $n-1$ jobs a fixed processing time, and placing the job with the maximal length $OPT$ and a rate $\epsilon$, as a single job on a machine in the optimal schedule.

\textbf{Case 2:} $b_n\rightarrow0$.

$C_{max}(\sigma)(0)\le[-R_n\cdot \ln(1-\frac{OPT}{R_n})]=OPT\cdot [-\frac{R_n}{OPT}\cdot \ln(1-\frac{OPT}{R_n})]$.
Setting $v=\frac{R_n}{OPT}$, the maximum of the term $-v\cdot \ln(1-1/v)$ is $\frac{e}{e-1}$.

Combining both limits, we get $$PoA(G^{LBDR})=\frac{C_{max}(\sigma)}{OPT}\le \max\{\frac{e}{e-1},2-\frac{1}{m}\}.$$

As stated in the proof, we assume that the last job in the priority list, job $n$, does not reach the threshold processing time $\tau$. The assumption is required since we use the linear relation between the actual processing time of the job, $p_n$, and its starting time $S_n$.
We allow earlier jobs in the priority list to reach $\tau$, since we only assume that their actual processing time is not larger than their basic processing time, which is true both in the linear regime of their processing time function, and after it. We conjecture that the theorem is true even if the last job reaches $\tau$, and we leave the formal proof for future research.
\hfill \qed

\posA
\noindent{\bf Proof of Theorem~\ref{thm:SBPTCombined}:}
Recall that with a global priority list, in particular SBPT, Algorithm List-Scheduling produces a $NE$. Moreover, for every $NE$, there exists a run of LS that produces it (Theorem~\ref{thm:ls_on_global}).

For a profile $\sigma$, let $P(\sigma)$ denote the total processing time of the jobs in $\sigma$. We first show that $P(\sigma)$ is uniform for all $NE$ profiles produced by LS. This property is valid for any global priority list. Next, we show that with SBPT policy $P(\sigma)$ is minimized in a $NE$ profile. Finally, we conclude the $PoA$ by analyzing the relation between the Makespan and the total processing time. 

\begin{claim}
\label{cl:uniformP}
For every game with a global priority list, the total jobs' processing time is uniform in all Nash Equilibria.

\end{claim}
\begin{proof}
For $1 \le i \le n$, let $L^i=(L_1^i,\ldots,L_m^i)$ be the sorted vector of machines' loads after job $i$ is assigned. That is, $L_1^i$ is the load on the most loaded machines, etc.
We show that $L_i$ is uniform in any application of List-Scheduling. The proof is by induction on $i$.
Initially, $L_0=(0,\ldots,0)$. For $i>0$, the $i$-th job is assigned on a machine with load $L_m^{i-1}$. Independent of the choice of this machine (in case of a tie-breaking), the vector $L^i$ includes one less machine with load $L_m^{i-1}$ and one more machine with load $(1+a)L_m^{i-1}+b_i$.

\end{proof}
\begin{lemma}
\label{lem:sumP}
In a game with SBPT policy, $NE$ profiles minimize the total jobs processing times.
\end{lemma}
\begin{proof}
    We first characterize the structure of a schedule that minimizes the total processing time, and then show that LS with a specific tie-breaking rule, produces a schedule that agrees with this structure.  Throughout the proof, we assume that $n=zm$ for an integer $z$. This is w.l.o.g., since at most $m-1$ jobs with $b_i=0$ can be added to $N$. With SBPT policy, these jobs are first in $\pi$. Both OPT and LS assign these jobs as first on some machines, they all complete at time $0$ and do not delay other jobs.  

\begin{claim}
\label{cl:rounds}
Let $\sigma$ be a profile that obeys SBPT, in which, for all $k=1,\ldots,z$, the jobs $(m-1)k+1,....,mk$ form the $k$-th round, then $\sigma$ minimizes the total processing time. 
\end{claim}
\begin{proof}
We show by induction on $j$, that for every $j$, the total processing time of any subset of $mj$ jobs is minimized if the set includes the first $mj$ jobs and they are processed in $j$ rounds such that for every $k<j$, the jobs $(m-1)k+1,\ldots,mk$ form the $k$-th round.

The base case considers an assignment of $m$ jobs. Clearly, 
assigning a single job on every machine is surely optimal for the first $m$ jobs.

Assume optimality for $j-1$. That is, the total processing time of $m(j-1)$ jobs is minimized by arranging the first $m(j-1)$ jobs in $\pi$ in $j-1$ rounds.
Consider a schedule of $mj$ jobs.
If there is an empty machine, then some machine has at least two jobs and the total processing time will not increase if we assign on the empty machine some not-first job. Since $n>m$, a not-first job exists. Therefore, we can assume that in a schedule that minimizes the total processing time of $mj$ jobs, there is at least one job on each machine. 

For a set of jobs $S$, denote by $B(S)$ the total basic processing time of the jobs in $S$. For a set of jobs $S$ and a profile $\sigma$, denote by $P(S,\sigma)$ the total processing time of the jobs in $S$ in the schedule $\sigma$.  
Let $J_{last}$ be the set of the $m$ jobs that are last on each of the $m$ machines. 
Let $J_{not last}$ be the set of other $m(j-1)$ jobs.
Note that the jobs in $J_{not last}$ delay the jobs in $J_{last}$, such that the total starting times of the jobs in $J_{last}$ is $P(J_{not last},\sigma)$.
Therefore, 
$$P(J_{last} \cup J_{not last},\sigma) = (1+a)\cdot P(J_{not last},\sigma)+ B(J_{last}).$$
Since $a>0$, this term is minimized by minimizing the total processing time of the $m(j-1)$ jobs in $J_{notlast}$, and selecting to the $j$-th round the shortest available jobs.
By the induction hypothesis, the total processing time of $J_{notlast}$ is minimized by placing the first $m(j-1)$ jobs in rounds.
This concludes the induction step.
\end{proof}

Consider an application of List-Scheduling in which, whenever a job has more than one machine that minimizes its completion time, it would prefer a one with fewer jobs on it. 
We show that a run of LS with the above tie-breaking rule produces a schedule that assigns the jobs in rounds, as described in Claim~\ref{cl:rounds}.

\begin{claim}
In a game with SBPT policy, if LS is applied with tie-breaking in favor of machines with fewer jobs, then for all $k=1,\ldots,z$, the jobs $(m-1)k+1,\ldots,mk$ form the $k$-th round.
\end{claim}
\begin{proof}
Assume towards contradiction that the claim is false and let $k$ be the smallest $k$ for which LS assigns at least two jobs among jobs $(k-1)m+1,\ldots,km$ on the same machine, $M_2$. 
There must be a machine, $M_0$, that is not assigned any job among jobs $(k-1)m+1,\ldots,kn$.
Let $A_0$ and $A_2$ be the loads on $M_0$ and $M_2$ after $k-1$ rounds. 
The tie-breaking rule ensures that the claim holds for the first round, even if there are dummy jobs of length $0$, therefore, $k>1$. 

Let $w$ be the last job on $M_0$ after $k-1$ rounds. Since $k>1$, such a $w$ exists.
Job $w$ is assigned on $M_0$ because $A_0 \le (1+a)A_2+b_w$.
Let $x,y$ be the first two jobs among jobs $(k-1)m+1,\ldots,kn$,  assigned on $M_2$. After the assignment of $x$, the load on $M_2$ is $(1+a)A_2+b_x$. 

Consider the assignment of job $y$. The number of jobs already assigned on $M_2$ and $M_0$ is $k$ and $k-1$, respectively. By the tie-breaking rule, $y$ should be assigned on $M_0$ is case of a tie. It is assigned on $M_2$ since $(1+a)[(1+a)A_2+b_x]+b_y < (1+a)A_0+b_y,$
which implies $A_0 > (1+a)A_2+b_x$. 

Combining this with the fact that $A_0 \le (1+a)A_2+b_w$, we conclude that $b_x < b_w$. However, since $w$ was assigned in round $k-1$ and $x$ was assigned in round $w$, the SBPT order implies that $b_w \le b_x$. A contradiction. 
\end{proof}

The statement of the Lemma follows from combining the above claims with Claim~\ref{cl:uniformP}.
\end{proof}
We can now complete the analysis of the Price of Anarchy.
Let $\sigma$ be a $NE$ profile.
Let $C_{max}(\sigma)$ denote the makespan of $\sigma$. Let $L_j$ denote the load on machine $j$, and let $L_{min}= \min_j L_j$ be the load on a least loaded machine. Let $n$ be the job that finishes last and determines $C_{max}(\sigma)$. 
Since $\sigma$ is a $NE$, job $n$ would not benefit from deviating to the machine with the minimal load in $\sigma$, hence, $C_{max}(\sigma)\le(1+a)\cdot L_{min} + b_n$.
We first bound the sum of jobs' processing times in $\sigma$:

$$\sum_ip_i^\sigma\ge m\cdot L_{min} +a\cdot L_{min} +b_n=(m+a)\cdot L_{min}+b_n.$$

Hence, 
$$\sum_ip_i^\sigma\ge \frac{m+a}{1+a}\cdot C_{max}(\sigma) +\frac{1-m}{1+a}\cdot b_n.$$

And,
$$C_{max}(\sigma)\le\frac{1+a}{m+a}\cdot \sum_ip_i^\sigma +\frac{m-1}{m+a}\cdot b_n$$

Now, consider the optimal schedule $\sigma^*$.
Clearly $OPT\ge b_n$.
In addition, by Lemma~\ref{lem:sumP}, the total jobs' length in OPT is at least $\sum_ip_i^{\sigma}$, therefore,  $OPT\ge \frac{\sum_ip_i^{\sigma}}{m}$.
We conclude that,
$$C_{max}(\sigma)\le \frac{(1+a)\cdot m}{m+a}\cdot \frac{\sum_ip_i^\sigma}{m} +\frac{m-1}{m+a}\cdot b_n \le (\frac{(1+a)\cdot m}{m+a}+\frac{m-1}{m+a})\cdot OPT = (\frac{2m+a\cdot m-1}{m+a})\cdot OPT.$$
\hfill \qed

The following example shows that this bound is tight for $m=2$ and $a=1$.  Consider a game $G$ with two identical machines $M_1$ and $M_2$, and three jobs $u,v,w$ with deterioration rate $a=1$. Let $b_u=b_v=1$ and $b_w=3$. A priority list that respects the $SBPT$ coordination mechanism is $\pi=(u,v,w)$.
    In an optimal schedule $\sigma^*$, (see Figure~\ref{fig:sbptPOA5/3}), jobs $u$ and $v$ are assigned on $M_1$ and job $w$ is assigned on $M_2$. The processing time of job $v$ when assigned after job $u$ is $1+1=2$, and hence $C_{max}(\sigma^*)=3$.
    A $NE$ profile assigns jobs $u$ and $v$ on machines $M_1$ and $M_2$ respectively, and job $w$ as the last job in one of the machines. 
    The processing time of job $w$ when assigned at time $t=1$ is $1+3=4$, and hence $C_{max}(\sigma)=5$.
    We conclude that $$PoA(G)=\frac{C_{max}(\sigma)}{C_{max(\sigma^*)}}=\frac{5}{3}=\frac{4+2-1}{2+1}.$$

\begin{figure}[h]
\centering
\includegraphics[width=0.6\textwidth]{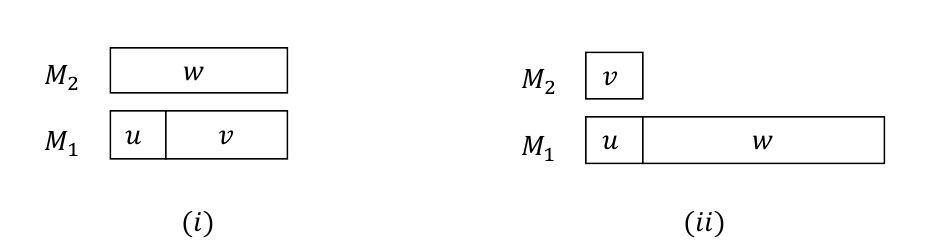}
\caption{$(i)$ An optimal profile with makespan 3. $(ii)$ A $NE$ profile with makespan 5.}
\label{fig:sbptPOA5/3}
\end{figure}


\begin{thebibliography}{10}

\bibitem{AlidaeeWormerSurvey}
B.~Alidaee and N.~K. Womer.
\newblock Scheduling with time dependent processing times: Review and extensions.
\newblock {\em The Journal of the Operational Research Society}, 50:711--720, 1999.

\bibitem{AD+08}
E.~Anshelevich, A.~Dasgupta, J.~Kleinberg, E.~Tardos, T.~Wexler, and T.~Roughgarden.
\newblock The price of stability for network design with fair cost allocation.
\newblock {\em SIAM Journal on Computing}, 38(4):1602--1623, 2008.

\bibitem{Aspnes}
J.~Aspnes, Y.~Azar, A.~Fiat, S.~Plotkin, and O.~Waarts.
\newblock On-line load balancing with applications to machine scheduling and virtual circuit routing.
\newblock {\em Proceedings of the Twenty-Fifth Annual ACM Symposium on Theory of Computing,}, (STOC), 1993.

\bibitem{AART06}
B.~Awerbuch, Y.~Azar, Y.~Richter, and D.~Tsur.
\newblock Tradeoffs in worst-case equilibria.
\newblock {\em Theoretical Computer Science}, 361(2):200--209, 2006.


\bibitem{AJM15}
Y.~Azar, L.~Fleischer, K.~Jain, V.~Mirrokni, and Z.~Svitkina.
\newblock Optimal coordination mechanisms for unrelated machine scheduling.
\newblock {\em Operations Research}, 63(3):489--500, 2015.

\bibitem{BJ00}
A.~Bachman and A.~Janiak.
\newblock Scheduling jobs with decreasing processing times for the total completion time minimization.
\newblock {\em Operations Research Proceedings}, 2000:353--358, 2000.

\bibitem{BrowneYechieaRandom}
S.~Browne and U.~Yechiali.
\newblock Scheduling deteriorating jobs on a single processor.
\newblock {\em Opns Res}, 38:495--498, 1990.


\bibitem{CLTY17}
Q.~Chen, L.~Lin, Z.~Tan, and Y.~Yan.
\newblock Coordination mechanisms for scheduling games with proportional deterioration.
\newblock {\em European Journal of Operational Research}, 263(2):380--389, 2017.

\bibitem{CKN09}
G.~Christodoulou, E.~Koutsoupias, and A.~Nanavati.
\newblock Coordination mechanisms.
\newblock {\em Theor. Comput. Sci.}, 410(36):3327--3336, 2009.


\bibitem{CzumajV07}
A.~Czumaj and B.~V\"{o}cking.
\newblock Tight bounds for worst-case equilibria.
\newblock {\em ACM Trans. Algorithms}, 3(1):4:1--4:17, 2007.


\bibitem{FST17}
M.~Feldman, Y.~Snappir, and T.~Tamir.
\newblock The efficiency of best-response dynamics.
\newblock In {\em The 10th International Symposium on Algorithmic Game Theory (SAGT)}, 2017.

\bibitem{GLM10}
M.~Gairing, T.~Lücking, and M.~Mavronicolas et~al.
\newblock Computing nash equilibria for scheduling on restricted parallel links.
\newblock {\em Theory of Computing Systems}, 47(2):405--432, 2010.

\bibitem{Gra66}
R.L. Graham.
\newblock Bounds for certain multiprocessing anomalies.
\newblock {\em Bell Systems Technical Journal}, 45:1563--1581, 1966.

\bibitem{GuptaGuptaSingel}
J.~N.~D.~Gupta and S.~K.~Gupta.
\newblock Single facility scheduling with nonlinear processing times.
\newblock {\em Computers ind. Engng}, 14:387--393, 1987.

\bibitem{HoDecreasingWithDeadlines}
K.~I-J.~Ho, J.~Y-T.~Leung, and W-D.~Wei.
\newblock Complexity of scheduling tasks with time dependent execution times.
\newblock {\em Information Processing Letters}, 48:315--320, 1993.

\bibitem{ILMS09}
N.~Immorlica, L.~Li, V.~S.~Mirrokni, and A.~S.~Schulz.
\newblock Coordination mechanisms for selfish scheduling.
\newblock {\em Theor. Comput. Sci.}, 410(17):1589--1598, 2009.

\bibitem{KangNgParallel}
L.~Kang and C.T.~Ng.
\newblock A note on fully polynomial-time approximation scheme for parallel-machine scheduling with deteriorating jobs.
\newblock {\em Int. J. Production Economics}, 109:180--184, 2007.

\bibitem{Koutsoupias:1999:WE:1764891.1764944}
E.~Koutsoupias and C.~Papadimitriou.
\newblock Worst-case equilibria.
\newblock {\em STACS'99: Proceedings of the 16th annual conference on Theoretical aspects of computer science}, 1:404--413, 1999.

\bibitem{LiLiuLiFirstGameTheory}
K.~Li, C.~Liu, and K.~Li.
\newblock An approximation algorithm based on game theory for scheduling simple linear deteriorating jobs.
\newblock {\em Theoretical Computer Science}, 543:46--51, 2014.

\bibitem{MosheiovFixedBasicVshapred}
G.~Moesheiov.
\newblock V-shaped policies for scheduling deteriorating jobs.
\newblock {\em Operations Research}, 39:979--991, 1991.

\bibitem{MoeshiovRankBased}
G.~Mosheiov.
\newblock A note on scheduling deteriorating jobs.
\newblock {\em Mathematical and Computer Modeling}, 41:883--886, 2005.

\bibitem{RST21}
V.~Ravindran Vijayalakshmi, M.~Schr{\"o}der, and T.~Tamir.
\newblock Scheduling games with machine-dependent priority lists.
\newblock {\em Theoretical Computer Science}, 855:90--103, 2021.

\bibitem{Ros73}
R.~W. Rosenthal.
\newblock A class of games possessing pure-strategy nash equilibria.
\newblock {\em International Journal of Game Theory}, 2:65--67, 1973.

\bibitem{3DM}
V.~Kann.
\newblock Maximum bounded 3-dimensional matching is max snp-complete.
\newblock {\em Inf. Process. Lett}, 37:27--35, 1991.

\bibitem{WangSquared}
J.~Wang.
\newblock A note on single-machine scheduling with decreasing time-dependant job processing times.
\newblock {\em Applied Mathematical Modelling}, 34:294--300, 2010.

\bibitem{WangXiaSpecialCase}
J.~Wang and Z.~Xia.
\newblock Scheduling jobs under decreasing linear deterioration.
\newblock {\em Information Processing Letters}, 94:63--69, 2005.

\bibitem{YuBatchProcessing}
G.~Yu.
\newblock A coordination mechanism for scheduling game on parallel-batch machines with deterioration jobs.
\newblock {\em Mathematical Problems in Engineering}, 2022.

\end{thebibliography}
\end{document}